\documentclass[11p, reqno]{amsart}
\usepackage{amssymb,amsthm,amsmath,mathrsfs,bm,braket,marginnote}
\numberwithin{equation}{section}
\pdfoutput=1
\usepackage{appendix} 
\usepackage[utf8]{inputenc}
\usepackage[T1]{fontenc}
\usepackage{microtype}

\usepackage{fourier}
\usepackage{caption}
\usepackage[colorlinks=true, pdfstartview=FitV, linkcolor=blue, citecolor=blue, urlcolor=blue]{hyperref}
\usepackage[numbers,sort&compress]{natbib}

\theoremstyle{plain}
\newtheorem{proposition}{Proposition}[section]
\newtheorem{corollary}[proposition]{Corollary}
\newtheorem{coro}[proposition]{Corollary}
\newtheorem{lemma}[proposition]{Lemma}
\newtheorem{theorem}[proposition]{Theorem}
\newtheorem{property}[proposition]{Property}

\theoremstyle{definition}

\newtheorem{remark}[proposition]{Remark}

\usepackage[usenames,dvipsnames]{xcolor}
\usepackage{pgf}
\usepackage{pgfplots}
\usepackage{tikz}
\usetikzlibrary{arrows,calc}
\usepackage{verbatim}
\usetikzlibrary{decorations.pathreplacing,decorations.pathmorphing}
\usepackage{multicol}  
\usepackage{multirow}
\usetikzlibrary{%
    decorations.pathreplacing,%
    decorations.pathmorphing%
}

\newcommand{\R}{\mathbf{R}}

 \reversemarginpar
\begin{document}

\title[Isospectral flows related to FST polynomials]{Isospectral flows related to Frobenius-Stickelberger-Thiele polynomials}
\author{Xiang-Ke Chang}
\address{ LSEC, ICMSEC, Academy of Mathematics and Systems Science, Chinese Academy of Sciences, P.O.Box 2719, Beijing 100190, PR China; and School of Mathematical Sciences, University of Chinese Academy of Sciences, Beijing 100049, PR China.}
\email{changxk@lsec.cc.ac.cn}
\author{Xing-Biao Hu}
\address{ LSEC, ICMSEC, Academy of Mathematics and Systems Science, Chinese Academy of Sciences, P.O.Box 2719, Beijing 100190, PR China; and School of Mathematical Sciences, University of Chinese Academy of Sciences, Beijing 100049, PR China.}
\email{hxb@lsec.cc.ac.cn}

\author{Jacek Szmigielski}
\address{Department of Mathematics and Statistics, University of Saskatchewan, 106 Wiggins Road, Saskatoon, Saskatchewan, S7N 5E6, Canada.}
\email{szmigiel@math.usask.ca}

\author{Alexei Zhedanov}
\address{School of Mathematics, Renmin University of China, Beijing 100872, China.}
\email{zhedanov@ruc.edu.cn}

\date{}

\dedicatory{}


\begin{abstract} 
The isospectral deformations of the Frobenius-Stickelberger-Thiele (FST) polynomials introduced in \cite{spiridonov2007integrable} (Spiridonov et al. Commun. Math. Phys. 272:139--165, 2007 ) are studied. For a specific choice of the deformation of the spectral measure, one is led to an integrable lattice (FST lattice), which is indeed an isospectral flow connected with a generalized eigenvalue problem.
In the second part of the paper the spectral problem used previously in the study of the modified Camassa-Holm (mCH) peakon lattice is interpreted in terms of the FST polynomials together with the associated FST polynomials, resulting in a map from the mCH peakon lattice to a negative flow of the finite FST lattice.  Furthermore, it is pointed out that the degenerate case of the finite FST lattice unexpectedly maps to the interlacing peakon ODE system associated with the two-component mCH equation studied in \cite{chang2016multipeakons} (Chang et al. Adv. Math. 299:1--35, 2016).  
\end{abstract}

\maketitle
\tableofcontents
\section{Introduction}
The Toda lattice is a celebrated completely integrable model for a one-dimensional crystal in solid state physics. The links between the theory of orthogonal polynomials (OPs) and the integrable systems of Toda type have been extensively investigated and used by both, i.e. integrable systems and special functions, communities since the early 1990s. In the semi-discrete case (continuous time and discrete space), one of the well-known examples is the connection between the ordinary OPs and the Toda lattice. The ordinary OPs appear as wave functions of the Lax pair of the semi-discrete Toda lattice undergoing a one-parameter deformation of the spectral measure \cite{chu2008linear,deift2000orthogonal,nakamura2004special,peherstorfer2007toda}.  As a second example,  the semi-discrete Lotka-Volterra lattice (sometimes also called the Kac--van Moerbeke lattice or the Langmuir lattice)  \cite{chang2016multipeakons,chu2008linear,kac1975explicitly} can be obtained as a one-parameter deformation of the measure associated to symmetric OPs.  For more examples, one may refer \cite{adler1995matrix,adler1997string,adler1999generalized,alvarez2013orthogonal,alvarez2011multiple,ariznabarreta2016multivariate,aptekarev1997toda,chang2018partial,kharchev1997faces,mukaihira2002schur,nenciu2005lax,peherstorfer2007toda,vinet1998integrable} etc. 

Later, it was found that spectral transformation of the ordinary orthogonal polynomials also play a central role in the solution of the peakon problem of the Camassa-Holm (CH) equation \cite{beals2000multipeakons}. Here we remark that the CH equation is a completely integrable system which arises as a shallow water wave model, and peakon solutions (simply called peakons) appear as its solitary wave solutions with peaks, whose dynamics can be described by a system of ODEs. The mathematics of peakons has attracted a great deal of attention.   Interestingly, the CH peakon dynamical system (i.e. the ODE system describing the CH peakons) and the finite Toda lattice are associated with different spectral transformations of the ordinary orthogonal polynomials and may be regarded as opposite flows in some sense \cite{beals2001peakons,ragnisco1996peakons}. Recently, more such relations of this type were discovered \cite{chang2018degasperis,chang2018application,chang2016multipeakons}.

As for full-discrete (discrete space and time) integrable systems, the compatibility of discrete spectral transformations of OPs may yield the full-discrete Toda lattice \cite{papageorgiou1995orthogonal,spiridonov1995discrete}.  The compatibility of discrete spectral transformations of symmetric OPs gives the full-discrete Lotka-Volterra lattice \cite{spiridonov1997discrete}.  Sometimes the full-discrete Toda lattice is also called the qd algorithm \cite{chang2015about,rutishauser1954quotienten}, which can be used to compute eigenvalues of a class of tridiagonal matrices.    Furthermore, the full-discrete Lotka-Volterra equation can be used as an efficient algorithm to compute singular values of certain tridiagonal matrix \cite{tsujimoto2001discrete}. Here we also mention that one step of the QR algorithm is equivalent to the time evolution of the finite semi-discrete Toda lattice \cite{symes1982qr,deift1983}. 


Based on the pioneering study by three mathematicians Frobenius, Stickelberger and Thiele, Spiridonov, Tsujimoto and Zhedanov introduced the notion of the Frobenius-Stickelberger-Thiele (FST)  polynomials \cite{spiridonov2007integrable}.  By constructing spectral transformations for these polynomials analogous to the Christoffel and Geronimus transformations for orthogonal polynomials, they proposed an integrable discrete time chain for the FST polynomials \cite{spiridonov2007integrable}. It is interesting that this chain is related to the generalized $\varepsilon$-algorithm, which is a convergence acceleration algorithm. However,  the 
question of the continuous time lattices associated with the FST polynomials still remains open. This paper is devoted to solving this problem.

More precisely,  in this paper, we are interested in the FST polynomials \cite{spiridonov2007integrable} and the related integrable continuous time lattices. The layout of the paper is as follows. In Section \ref{sec:fstpoly}, we introduce a family of polynomials given by  explicit determinantal formulae. We prove that these polynomials satisfy a biorthogonality relation as well as the three term recurrence relation known from the theory of the FST polynomials. Then, in Section \ref{sec:inf_fst}, by considering a time-dependent measure, we derive the evolution of these FST polynomials leading to an  integrable lattice of infinite dimension.  The finite truncation of such an FST lattice is introduced in Section \ref{sec:finite_fst} while its relation with the modified Camassa-Holm (mCH) peakon lattice \cite{chang2017lax} is established in Section \ref{sec:mch_fst}. Finally, in Section \ref{sec:fst_2mch}, we discuss a degenerate case of the FST lattice and associate it with the interlacing peakon ODE system of a two-component mCH (2-mCH) equation studied in \cite{chang2016multipeakons}.
It should be mentioned that the FST lattice is an isospectral flow connected with a generalized eigenvalue problem. To the best of our knowledge, this provides another example of such an isospectral flow other than the relativistic Toda chain \cite{kharchev1997faces} and the R-I chain \cite{vinet1998integrable}. Besides, the role of associated FST polynomials is highlighted in order to interprete the spectral problem of the mCH peakon lattice.

\subsection{Notations} \label{sec:notations}

For convenience, we use throughout the paper the notation of \cite{chang2017lax}.  
\begin{enumerate}
\item $\binom{[k]}{j}$ denotes the set of all $j$-element subsets of $[k]$, listed in increasing order; for example $I\in \binom{[k]}{j}$ means that 
$I=\{i_1, i_2,\dots, i_j\}$ for some increasing sequence $i_1 < i_2 < \dots < i_j\leq ~k$;  and we use the convention
$\binom{[k]}{0}=1; \ \binom{[k]}{j}=0,\ \  j>k.$
\item  Given the multi-index $I$, we abbreviate $g_I=g_{i_1}g_{i_2}\dots g_{i_j}$ etc.
\item Let $I,J \in \binom{[k]}{j}$, or $I\in \binom{[k]}{j+1},J \in \binom{[k]}{j}$.  Then  $I, J$ are said to be \emph{interlacing} if 
\begin{equation*}
  \label{eq:interlacing}
    i_{1} <j_{1} < i_{2} < j_{2} < \dotsb < i_{j} <j_{j}
\end{equation*}
or, 
\begin{equation*}
    i_{1} <j_{1} < i_{2} < j_{2} < \dotsb < i_{j} <j_{j}<i_{j+1}, 
\end{equation*}
in the latter case.  
We abbreviate this condition as $I < J$ in either case, and, furthermore, 
use the same notation, that is $I<J$,   for $I\in \binom{[k]}{1}, J \in \binom{[k]}{0}$.

\item For two ordered multi-index sets $I, J$,  we define 
\begin{align*}
  \mathbf{x}_J&=\prod_{j\in J}x_j, &\Delta_J(\mathbf{x})&=\prod_{i<j\in J}(x_j-x_i), \\
  \Delta_{I,J}(\mathbf{x};\mathbf{y})&=\prod_{i\in I}\prod_{j\in J}(x_i-y_j), 
  &\Gamma_{I,J}(\mathbf{x};\mathbf{y})&=\prod_{i\in I}\prod_{j\in J}(x_i+y_j),  
  \end{align*}
  along with the convention
\begin{align*}
&\Delta_\emptyset(\mathbf{x})=\Delta_{\{i\}}(\mathbf{x})=\Delta_{\emptyset,J}(\mathbf{x};\mathbf{y})=\Delta_{I,\emptyset}(\mathbf{x};\mathbf{y})=\Gamma_{\emptyset,J}(\mathbf{x};\mathbf{y})=\Gamma_{I,\emptyset}(\mathbf{x};\mathbf{y})=1.
\end{align*}

\end{enumerate}

Furthermore, we introduce a generalization of the Cauchy-Vandermonde matrix.
Let $\{e_k\}_{k=0}^\infty$ be a sequence of real numbers such that the numbers $e_k$ are distinct, positive, i.e. $0<e_k \neq e_j$ for $k\neq j$. Given  an index $l$ such that $0\leq l<k$, another index $p$ such that $0\leq p$ and a positive measure $\nu$ 
 with support in $\R_+$, a \textit{ Cauchy-Stieltjes-Vandermonde (CSV) matrix} is a matrix of the form 
\begin{equation}\label{CSV}
C_k^{(l,p)}(\nu,\bf{e})=\left(\begin{array}{cccccccc}
    e_1^{p}V(e_1)&e_1^{p+1} V(e_1)&\cdots&e_1^{p+l-1}V(e_1)&1&e_1&\cdots&e_1^{k-l-1}\\
    e_2^{p}V(e_2)&e_2^{p+1}V(e_2)&\cdots&e_2^{p+l-1}V(e_2)&1&e_2&\cdots&e_2^{k-l-1}\\
    \vdots&\vdots&\ddots&\vdots&\vdots&\vdots&\ddots&\vdots\\
    e_k^{p}V(e_k)&e_k^{p+1}V(e_k)&\cdots&e_k^{p+l-1}V(e_k)&1&e_k&\cdots&e_k^{k-l-1}
  \end{array}\right),
\end{equation}
where $ V$ is the Stieltjes transform of the measure 
$\nu$ given by
$
 V(z)=\int\frac{d\nu(x)}{x+z}. 
$ Henceforth, we will suppress in the notation the dependence on $\nu, \bf{e}$.  
 
If, in addition, $p+l-1\leq k-l$ one can prove that the determinant of 
$C_k^{(l,p)}$ admits the multiple integral representation \cite{chang2017lax}
\begin{equation} \label{eq:detCSV1}
\det \left(C_k^{(l,p)}\right)=(-1)^{lp+\frac{l(l-1)}{2}} \Delta_{[1,k]}(\mathbf{e})
\idotsint\limits_{0<x_1<x_2<\dots<x_l} \frac{\Delta_{[1,l]}(\mathbf{x})^2}{\Gamma_{[1,k], [1,l]}(\mathbf{e}; \mathbf{x})} d\nu^p(x_1)d\nu^p(x_2)\dots d\nu^p(x_l),
\end{equation} 
where $d\nu^p(x)=x^pd\nu(x)$.

%

\section{FST polynomials}\label{sec:fstpoly}

The pioneering study of the FST polynomials was carried out by three mathematicians; the three term recurrence relation characterizing the FST polynomials appeared first in the paper \cite{frobenius1880uber} by Frobenius and Stickelberger, devoted to elliptic functions, and later in the work of Thiele on the rational interpolation problems \cite{thiele1909interpolation}. Thus the context relevant to these polynomials is associated with all three names  as was pointed out in \cite{spiridonov2007integrable} and we will use that terminology throughout the paper.  

We shall start from the definition based on an explicit formula discussed in \cite{chang2017lax} rather than  based on the three term recurrence relation presented in \cite{spiridonov2007integrable}. Thus, let us  consider a family of polynomials  $\left\{T_k(z)\right\}_{k=0}^\infty$ given by 
\begin{align}\label{fst_poly}
T_{k}(z)=\frac{1}{N_k}\det\left(
\begin{array}{cccccccc}
    1&-z&\cdots&(-z)^{\lfloor \frac k2 \rfloor}&0&0&\cdots&0\\
    V(e_1)&e_1V(e_1)&\cdots&e_1^{\lfloor \frac k2 \rfloor}V(e_1)&1&e_1&\cdots&e_1^{{\lfloor \frac {k-1}2 \rfloor}}\\
    V(e_2)&e_2V(e_2)&\cdots&e_2^{\lfloor \frac k2 \rfloor}V(e_2)&1&e_2&\cdots&e_2^{\lfloor \frac {k-1}2 \rfloor}\\
    \vdots&\vdots&\ddots&\vdots&\vdots&\vdots&\ddots&\vdots\\
    V(e_{k})&e_{k}V(e_{k})&\cdots&e_{k}^{\lfloor \frac k2 \rfloor}V(e_{k})&1&e_{k}&\cdots&e_{k}^{\lfloor \frac {k-1}2 \rfloor}
  \end{array}
\right),
\end{align}
with $T_0(z)=1, T_{-1}(z)=0$,
where
$$N_{2p}=\det \left(C_{2p}^{\left(p,0\right)}\right),\qquad N_{2p+1}=(-1)^p\det \left(C_{2p+1}^{\left(p+1,0\right)}\right).$$

Polynomials \eqref{fst_poly} posses the following elementary properties. 
\begin{property}
For $p=0,1,\ldots,$
\[\deg(T_{2p}(z)) = \deg(T_{2p+1}(z)) = p.\]
Polynomials $T_{2p}(z)$ are monic, while $T_{2p+1}(z)$ have the form
\begin{align*}
T_{2p+1}(z)=(-1)^p\frac{\det \left(C_{2p+1}^{\left(p,0\right)}\right)}{\det \left(C_{2p+1}^{\left(p+1,0\right)}\right)}z^p+\mathcal{O}(z^{p-1}),
\end{align*}
that is, the coefficient of the highest degree of $T_{2p+1}(z)$ is
$$T_{2p+1}^+=(-1)^p\frac{\det \left(C_{2p+1}^{\left(p,0\right)}\right)}{\det \left(C_{2p+1}^{\left(p+1,0\right)}\right)}.$$
In particular, $T_0(z)=1,\ T_1(z)=\frac{1}{V(e_1)}$.
\end{property}
\begin{proof}
The proof is elementary and we omit it.
\end{proof}
\begin{property} \label{prop_ortho}
The polynomials $\left\{T_k(z)\right\}_{k=0}^\infty$ satisfy the orthogonality relation
\begin{align}
\int\frac{T_k(z)z^j}{\prod_{i=1}^k(z+e_i)}d\nu(z)=\beta_k\delta_{ \lfloor k/2 \rfloor,j},\qquad j=0,1,\ldots, \lfloor \frac k2 \rfloor, \label{eq:orth}
\end{align}
where 
\[
\qquad \beta_{2p}=\int d\nu(z)+(-1)^{p}\frac{\det \left(C_{2p}^{\left(p+1,0\right)}\right)}{\det \left(C_{2p}^{\left(p,0\right)}\right)}, \qquad \beta_{2p+1}=1, \qquad p=0,1,\ldots, 
\]
with the {\sl{proviso}} that for $p=0$ the second summand in $\beta_{2p}$  is absent.  
\end{property}
\begin{proof}

For  fixed $j=0,1,\ldots, \lfloor \frac k2 \rfloor-1$, we have
\begin{align*}
&\int\frac{T_k(z)(-z)^j}{\prod_{i=1}^k(z+e_i)}d\nu(z)=\\
&\frac{1}{N_k}\det\overbrace{ \left(
\begin{array}{cccccccc}
    \int\frac{(-z)^j}{\prod_{i=1}^k(z+e_i)}d\nu(z)&\int\frac{(-z)^{j+1}}{\prod_{i=1}^k(z+e_i)}d\nu(z)&\cdots&\int\frac{(-z)^{j+\lfloor \frac k2 \rfloor}}{\prod_{i=1}^k(z+e_i)}d\nu(z)&0&0&\cdots&0\\
    V(e_1)&e_1V(e_1)&\cdots&e_1^{\lfloor \frac k2 \rfloor}V(e_1)&1&e_1&\cdots&e_1^{{\lfloor \frac {k-1}2 \rfloor}}\\
    V(e_2)&e_2V(e_2)&\cdots&e_2^{\lfloor \frac k2 \rfloor}V(e_2)&1&e_2&\cdots&e_2^{\lfloor \frac {k-1}2 \rfloor}\\
    \vdots&\vdots&\ddots&\vdots&\vdots&\vdots&\ddots&\vdots\\
    V(e_{k})&e_{k}V(e_{k})&\cdots&e_{k}^{\lfloor \frac k2 \rfloor}V(e_{k})&1&e_{k}&\cdots&e_{k}^{\lfloor \frac {k-1}2 \rfloor}
  \end{array}
\right)}^A.
\end{align*}
Introduce a function $f_l(z)=\frac{(-z)^l} {\prod_{i=1}^k (z+e_i)}, \, 0\leq l\leq 2\lfloor{\frac k2}\rfloor-1$.  This is a rational function on $\mathbf{C}$, vanishing at $z=\infty$, having 
simple poles only, so 
\begin{equation*}
\sum_{i=1}^k \frac{\textrm{Res}(f_l(-e_i))}{z+e_i}=f_l(z), 
\end{equation*}
or, written in integral form, 
\begin{equation*} 
\sum_{i=1}^k \textrm{Res}(f_l(-e_i)V(e_i))=\int f_l(z)d\nu (z). 
\end{equation*} 
The latter equation can be further simplified by observing 
\begin{equation*} 
\textrm{Res}(f_l(-e_i))=(e_i)^{l-j} \textrm{Res}(f_j(-e_i)), \quad 0\leq j\leq l, 
\end{equation*} 
from which 
\begin{equation}\label{eq:resintegral}
\sum_{i=1}^k  \textrm{Res}(f_j(-e_i))e_i^{l-j} V(e_i)=\int f_l(z)d\nu (z), \qquad 0\leq j\leq l\leq j+\lfloor \frac k2\rfloor 
\end{equation}
follows.  Moreover, the sum of 
its residues, including the residue at infinity, is $0$, resulting in 
\begin{equation}\label{eq:resformula}
\sum_{i=1}^k e_i^{l-j}  \textrm{Res}(f_j(-e_i))=-\textrm{Res}(f_l(\infty)).  
\end{equation}
In particular, 
\begin{equation*}
\sum_{i=1}^k e_i^{l-j}  \textrm{Res}(f_j(-e_i))=0,  \qquad 0\leq j\leq l\leq j+\lfloor 
\frac{k-1}{2} \rfloor, 
\end{equation*}
which proves, in conjunction with (\ref{eq:resintegral}), that the first row of $A$ is, as long as ~$0\leq~ j\leq ~\lfloor \frac k2 \rfloor-1$, a linear combination of 
the remaining rows.  

Now we turn to the case $j=\lfloor \frac k2\rfloor$.  
It is easier to do the analysis by considering separately even $k$ and odd $k$ respectively.    
Suppose $k=2p$.  
Then $f_{2p}(z)$ has a non-zero limit at $z=\infty$, namely $1$, and 
the correct partial fraction decomposition 
reads 
\begin{equation*}
1+\sum_{i=1}^{2p} e_i^p \frac{\textrm{Res}(f_p(-e_i))}{z+e_i}=f_{2p}(z). 
\end{equation*}
Then after integration we obtain 
\begin{equation*} 
\int d\nu(z)+\sum_{i=1}^{2p} e_i^p \textrm{Res}(f_p(-e_i)) V(e_i)=\int f_{2p}(z)d\nu(z), 
\end{equation*}
while for the remaining first $p$ columns in $A$, with $j=p=\lfloor \frac k2\rfloor$, (\ref{eq:resintegral}) is in force.  Likewise, we need to 
verify the residue computation based on (\ref{eq:resformula}) for 
all the columns in the Vandermonde part.  It is easy to check that only the last column will be affected.  Indeed
\begin{equation*} 
\sum_{i=1}^{2p} e_i^{p-1} \textrm{Res} f_p(-e_i)=-\textrm{Res}(f_{2p-1}(\infty))=-1, 
\end{equation*} 
and, after performing simple row operations, we obtain 
\begin{align*} 
&\det A=\det \begin{pmatrix} 0&0&\cdots&\int d\nu(z)&0&0&\cdots&1\\ 
V(e_1)&e_1V(e_1)&\cdots&e_1^p V(e_1)&1&e_1&\cdots&e_1^{p-1}\\
V(e_2)&e_2V(e_2)&\cdots&e_2^p V(e_2)&1&e_2&\cdots&e_2^{p-1}\\
\vdots&\vdots&\ddots&\vdots&\vdots&\vdots&\ddots&\vdots\\
V(e_{2p})&e_{2p}V(e_{2p})&\cdots&e_{2p}^p V(e_{2p})&1&e_{2p}&\cdots&e_{2p}^{p-1} \end{pmatrix}=\\
&(-1)^p \int d\nu(z)\, \det \left(C^{(p,0)}_{2p}\right)+\det \left (C^{(p+1,0)}_{2p}\right). 
\end{align*} 
Hence, for $k=2p$, we obtain 
\begin{equation*}
\int\frac{T_{2p}(z) (-z)^p }{\prod_{i=1}^{2p}(z+e_i)} d\nu(z)=
(-1)^p \int d\nu(z)+ \frac{\det\left(C^{(p+1,0}_{2p}\right)}{\det\left(
C^{(p,0)}_{2p}\right)}, 
\end{equation*}
which proves the claim for $k=2p$.  

For $k=2p+1$, on the other hand, we only need to verify the case $j=p, l=2p$ in 
(\ref{eq:resformula}), which gives 
\begin{equation*} 
\sum_{i=1}^{2p+1} e_i^p \textrm{Res}(f_p(-e_i))=1, 
\end{equation*} 
resulting in 
\begin{align*} 
&\det A=\det \begin{pmatrix} 0&0&\cdots&0&0&0&\cdots&-1\\ 
V(e_1)&e_1V(e_1)&\cdots&e_1^p V(e_1)&1&e_1&\cdots&e_1^{p}\\
V(e_2)&e_2V(e_2)&\cdots&e_2^p V(e_2)&1&e_2&\cdots&e_2^{p}\\
\vdots&\vdots&\ddots&\vdots&\vdots&\vdots&\ddots&\vdots\\
V(e_{2p+1})&e_{2p+1}V(e_{2p+1})&\cdots&e_{2p+1}^p V(e_{2p+1})&1&e_{2p+1}&\cdots&e_{2p+1}^{p} \end{pmatrix}=\\
&\det \left (C^{(p+1,0}_{2p+1}\right), 
\end{align*}
which shows 
\begin{equation*}
\int\frac{T_{2p+1}(z) (-z)^p }{\prod_{i=1}^{2p+1}(z+e_i)} d\nu(z)=
\frac{\det\left(C^{(p+1,0}_{2p+1}\right)}{\det\left(
(-1)^p C^{(p+1,0)}_{2p+1}\right)}=(-1)^p, 
\end{equation*}
thus completing the proof for $k=2p+1$. 

\end{proof}

\begin{property} \label{prop:3term}
The polynomials $\{T_k(z)\}_{k=0}^\infty$ satisfy the three term recurrence
\begin{align}\label{fst_3term}
T_{k+1}(z)=d_{k+1}T_{k}(z)+(z+e_k)T_{k-1}(z),
\end{align}
where 
\begin{subequations}
 \begin{align}
&d_{2p+1}=(-1)^p\frac{\det \left(C_{2p+1}^{\left(p,0\right)}\right)}{\det \left(C_{2p+1}^{\left(p+1,0\right)}\right)}-(-1)^{p-1}\frac{\det \left(C_{2p-1}^{\left(p-1,0\right)}\right)}{\det \left(C_{2p-1}^{\left(p,0\right)}\right)},\\
&d_{2p+2}=(-1)^{p+1}\frac{\det \left(C_{2p+2}^{\left(p+2,0\right)}\right)}{\det \left(C_{2p+2}^{\left(p+1,0\right)}\right)}-(-1)^{p}\frac{\det \left(C_{2p}^{\left(p+1,0\right)}\right)}{\det \left(C_{2p}^{\left(p,0\right)}\right)}
 \end{align}
 \end{subequations}
for $p=0,1,\ldots, $
with the {\sl{proviso}} that the second terms above are absent for $p=0$.  
\end{property}
\begin{proof}
We focus on the case $k\geq 2$ since the cases $k=0,1$ follow directly from the definition \eqref{fst_poly}.
Let us consider 
\begin{equation}
\frac{T_{k+1}(z)}{\prod_{i=1}^{k}(z+e_i)}=\frac{T_{k+1}(z)z}{\prod_{i=1}^{k+1}(z+e_i)}+\frac{T_{k+1}(z)e_{k+1}}{\prod_{i=1}^{k+1}(z+e_i)}=\displaystyle\sum_{j=0}^{k} \alpha_{k,j} \frac{T_j(z)}{\prod_{i=1}^{j}(z+e_i)},  \label{assum_T_expan}
\end{equation}
with the coefficients $\alpha_{k,j}$ to be determined and the convention that the empty product  (for $j=0$) in the denominator equals $1$.  
We start by carrying out the limit $z\rightarrow \infty$  to get
$
\alpha_{k,0}=0.
$
Then integrating both sides and using orthogonality (see \ref{prop_ortho}) we obtain
$\alpha_{k,1}=0. $
In the next step we multiply both sides of \eqref{assum_T_expan} by $z$ and carry out the limit $z \rightarrow \infty$ to get 
$
\alpha_{k,2}=0, 
$
while orthogonality implies $\alpha_{k,3}=0$.  
Repeating the operations of multiplying \eqref{assum_T_expan} by $z^l$
for $l=2,3,\ldots,\lfloor \frac {k}2 \rfloor-1$ one gets
$$\alpha_{k,2l}=\alpha_{k,2l-1}=0.$$

When $k=2p$ we obtain 
\[
\frac{T_{2p+1}(z)z}{\prod_{i=1}^{2p+1}(z+e_i)}+\frac{T_{2p+1}(z)e_{2p+1}}{\prod_{i=1}^{2p+1}(z+e_i)}= \alpha_{2p,2p} \frac{T_{2p}(z)}{\prod_{i=1}^{2p}(z+e_i)}+ \alpha_{2p,2p-1} \frac{T_{2p-1}(z)}{\prod_{i=1}^{2p-1}(z+e_i)}.
\]
 Again, multiplying by $z^{p-1}$, and integrating both sides using orthogonality, we obtain 
$$\alpha_{2p,2p-1}=1.$$
Finally, by comparing the leading terms on both sides we obtain 
$$\alpha_{2p,2p}=T_{2p+1}^+-T_{2p-1}^+=(-1)^p\frac{\det \left(C_{2p+1}^{\left(p,0\right)}\right)}{\det \left(C_{2p+1}^{\left(p+1,0\right)}\right)}-(-1)^{p-1}\frac{\det \left(C_{2p-1}^{\left(p-1,0\right)}\right)}{\det \left(C_{2p-1}^{\left(p,0\right)}\right)}.$$

When $k=2p+1$, proceeding as above, we obtain
\[
\frac{T_{2p+2}(z)z}{\prod_{i=1}^{2p+2}(z+e_i)}+\frac{T_{2p+2}(z)e_{2p+2}}{\prod_{i=1}^{2p+2}(z+e_i)}= \alpha_{2p+1,2p+1} \frac{T_{2p+1}(z)}{\prod_{i=1}^{2p+1}(z+e_i)}+ \alpha_{2p+1,2p} \frac{T_{2p}(z)}{\prod_{i=1}^{2p}(z+e_i)}, 
\]
 which upon comparing the leading terms on both sides implies 
  $$\alpha_{2p+1,2p}=1.$$
  Finally, multiplying by $z^{p}$ and integrating both sides and using  orthogonality, we have 
  $$ \alpha_{2p+1,2p+1}=\beta_{2p+2}-\beta_{2p}=(-1)^{p+1}\frac{\det \left(C_{2p+2}^{\left(p+2,0\right)}\right)}{\det \left(C_{2p+2}^{\left(p+1,0\right)}\right)}-(-1)^{p}\frac{\det \left(C_{2p}^{\left(p+1,0\right)}\right)}{\det \left(C_{2p}^{\left(p,0\right)}\right)}.$$
\end{proof}

The following relations follow readily.  
\begin{coro}\label{coro_sum}
 Let $d_0=\int d\nu(z)$.  Then, for $k=0,1,\ldots,$
\begin{subequations}
\begin{align}
&\sum_{i=0}^kd_{2i+1}=(-1)^k\frac{\det \left(C_{2k+1}^{\left(k,0\right)}\right)}{\det \left(C_{2k+1}^{\left(k+1,0\right)}\right)}=T_{2k+1}^+,\\
& \sum_{i=0}^kd_{2i}=d_0+(-1)^{k}\frac{\det \left(C_{2k}^{\left(k+1,0\right)}\right)}{\det \left(C_{2k}^{\left(k,0\right)}\right)}=\beta_{2k}.  
\end{align}
\end{subequations}
\end{coro}

We conclude this subsection by introducing the associated FST polynomials $\left\{T_k^{(1)}(z)\right\}_{k=0}^{\infty}$, defined by the same recurrence relation \eqref{fst_3term} 
\begin{align*}
T_{k+1}^{(1)}(z)=d_{k+1}T_{k}^{(1)}(z)+(z+e_k)T_{k-1}^{(1)}(z),
\end{align*}
but with different initial conditions, namely, 
$$T_{0}^{(1)}(z)=0,\qquad T_{1}^{(1)}(z)=1.$$

By an elementary result in the continued fraction theory, one immediately obtains 
\begin{property}
The ratio $\frac{T_{k}^{(1)}(z)}{T_k(z)}$ has the continued fraction expansion
\begin{align}
\frac{T_{k}^{(1)}(z)}{T_k(z)}=\cfrac{1}{d_1+\cfrac{z+e_1}{d_2+\cfrac{z+e_2}{d_3+\cdots+\cfrac{z+e_{k-2}}{d_{k-1}+\cfrac{z+e_{k-1}}{d_k}}}}}
\end{align}
\end{property}

\section{Infinite FST lattice}\label{sec:inf_fst}
In this section, we shall investigate a Toda-type dynamical system, henceforth called the TST lattice,  corresponding to a deformation of the FST polynomials.  To this end we introduce a simple ``time'' evolution of the measure $d\nu(z)$ given by the formula 
$$d\nu(z;t)=e^{zt}d\nu(z;0),$$
which yields a family of time-dependent FST polynomials $\left\{T_k(z;t)\right\}_{k=0}^\infty$. The following theorem gives the time evolution of 
rational functions $\left\{\psi_k(z;t)\stackrel{def}{=}\frac{T_k(z;t)}{\prod_{i=1}^k(z+e_i)}\right\}_{k=0}^\infty$, which we will refer to as the {\sl FST fractions}.  

\begin{theorem}
The FST fractions $\left\{\psi_k(z;t)\right\}_{k=0}^\infty$ corresponding to the time-dependent measure $d\nu(z;t)=e^{zt}d\nu(z;0)$ undergo the time evolution 
\begin{subequations}\label{lax_t}
\begin{align}
&\dot \psi_{2p+1}(z)=\left(-A_pB_p+e_{2p+1}\right)\psi_{2p+1}(z)+ A_{p-1}\psi_{2p}(z)-\psi_{2p-1}(z), \label{lax_t_odd}\\
\notag\\
&\dot \psi_{2p}(z)=A_{p-1}B_p\psi_{2p}(z)-B_p\psi_{2p-1}(z),  \label{lax_t_even}
\end{align}
\end{subequations}
where $A_p=\sum_{i=0}^{p}d_{2i+1},B_p=\sum_{i=0}^{p}d_{2i}$ 
and the dot means the derivative with respect to $t$.
\end{theorem}

\begin{proof}
We only present the proof for $\psi_k, k\geq2$ since the evolutions of $\psi_0,\psi _1$ easily follow from the definitions of $T_0$ and $T_1$.

Let us set, suppressing the time dependence, 
\begin{equation}
\dot \psi_{k}(z)=\displaystyle\sum_{j=0}^{k} \alpha_{k,j} \psi_j(z)\label{assum_Tdot_expan}
\end{equation}
with the coefficients $\alpha_{k,j}$ to be determined.

The limit $z\rightarrow \infty$ readily implies 
$
\alpha_{k,0}=0.
$
Now,  , we implement successively the following procedure (see also the proof of (\ref{fst_3term} )).  
Let us begin by taking the time derivative of the orthogonality relation $\int\psi_k(z)d\nu(z;t)=0$
which gives 
$$\int\dot \psi_k(z)d\nu(z;t)+\int\psi_k(z)zd\nu(z;t)=0.$$
Using \eqref{assum_Tdot_expan} and the orthogonality relation 
\eqref{eq:orth} we obtain $\beta_1 \alpha_{k,1}=0$ hence $\alpha_{k,1}=0$.  
Moreover, multiplying both sides of \eqref{assum_Tdot_expan} by $z$ and 
taking the limit $z\rightarrow \infty$ we get $\alpha_{k,2}=0$.  
We now repeat this procedure for $l=2,\ldots, \lfloor \frac {k}2 \rfloor-1$
\begin{enumerate}
\item Take the time derivative  of the orthogonality relation 
$\int\psi_k(z)z^{l-1}d\nu(z;t)=0$ to obtain 
$$\int\dot \psi_k(z)z^{l-1}d\nu(z;t)+\int\psi_k(z)z^{l}d\nu(z;t)=0.$$
By replacing $\dot T_k(z)$ in accordance with \eqref{assum_Tdot_expan}, 
with the proviso that the summation starts at $2l-1$ and employing the orthogonality relation \eqref{eq:orth}, we obtain $\alpha_{k,2l-1}=0$.

\item Multiply both sides of \eqref{assum_Tdot_expan}  by $z^l$ and take the limit $z\rightarrow \infty$ to obtain 
$$\alpha_{k,2l}=0.$$
\end{enumerate}
Consequently, we arrive at 
\[
\dot \psi_{k}(z)=\displaystyle\sum_{j=2\lfloor \frac {k}2 \rfloor-1}^{k} \alpha_{k,j} \psi_j(z).
\]
Henceforth, we proceed by treating the even and odd cases separately.

For $k=2p$, we have 
\[
\dot \psi_{2p}(z)= \alpha_{2p,2p} \psi_{2p}(z)+\alpha_{2p,2p-1} \psi_{2p-1}(z), 
\]
and implementing the first two steps above for $l=p$ results in 
\begin{align*}
&\alpha_{2p,2p-1}\int\psi_{2p-1}(z)z^{p-1}d\nu(z)+\int\psi_{2p}(z)z^{p}d\nu(z)=0,\\
&\alpha_{2p,2p}+\alpha_{2p,2p-1}T_{2p-1}^+=0,
\end{align*}
which yields
$$\alpha_{2p,2p-1}=-\beta_{2p}=-\sum_{i=0}^{p}d_{2i},\qquad   \alpha_{2p,2p}=-\alpha_{2p,2p-1}T_{2p-1}^+=\sum_{i=0}^{p}d_{2i} \sum_{i=0}^{p-1}d_{2i+1}.$$
This completes the proof for the even case \eqref{lax_t_even}.

For $k=2p+1$, we obtain 
\[
\dot \psi_{2p+1}(z)=\alpha_{2p+1,2p+1} \psi_{2p+1}(z)+ \alpha_{2p+1,2p} \psi_{2p}(z)+\alpha_{2p+1,2p-1} \psi_{2p-1}(z), 
\]
and upon implementing the first two steps above for $l=p$ we get
\begin{align*}
&\alpha_{2p+1,2p-1}\int\psi_{2p-1}(z)z^{p-1}d\nu(z)+\int\psi_{2p+1}(z)z^{p}d\nu(z)=0,\\
&\alpha_{2p+1,2p}+\alpha_{2p+1,2p-1}T_{2p-1}^+=0,
\end{align*}
which yield
$$\alpha_{2p+1,2p-1}=-1,\qquad   \alpha_{2p+1,2p}=-\alpha_{2p+1,2p-1}T_{2p-1}^+= \sum_{i=0}^{p-1}d_{2i+1}.$$
Likewise, implementing the first step above for $l=p+1$, we obtain
\begin{align*}
\alpha_{2p+1,2p+1}+\alpha_{2p+1,2p}\beta_{2p}+\alpha_{2p+1,2p-1}\int\psi_{2p-1}(z)z^{p}d\nu(z)+\int\psi_{2p+1}(z)z^{p+1}d\nu(z)=0, 
\end{align*}
which, with the help of the recurrence relation \eqref{fst_3term} applied to the last integral, implies 
\begin{align*}
\alpha_{2p+1,2p+1}=-\sum_{i=0}^{p}d_{2i+1}\sum_{i=0}^{p}d_{2i}+e_{2p+1}.
\end{align*}
This completes the proof for the odd case  \eqref{lax_t_odd}.
\end{proof}

Combining \eqref{fst_3term} and \eqref{lax_t}, we obtain an overdetermined system which we now present.  
\begin{lemma} \label{lem:Laxinfty}
Let 
\begin{equation*}
\Psi=(\psi_{0}(z;t),\psi_1(z;t),\cdots)^\top,\qquad \psi_k(z,t)=\frac{T_k(z;t)}{\prod_{i=1}^k(z+e_i)},
\end{equation*} 
then (\ref{fst_3term}) and (\ref{lax_t}) are equivalent to 

\begin{align}
L\Psi=zE\Psi,\qquad \dot \Psi=F R\Psi, \label{lax_inf_fst}
\end{align} 
where
\begin{align*}
&L=\left(\begin{array}{ccccccc}
d_1&-e_1&\\
1&d_2&-e_2\\
&1&d_3&-e_3\\
&&\ddots&\ddots&\ddots\\
&&&&\ddots
\end{array}
\right),
\qquad\qquad 
E=\left(\begin{array}{ccccccc}
0&1\\
&0&1\\
&&0&1\\
&&&\ddots&\ddots\\
&&&&\ddots
\end{array}
\right), \qquad F=E^\top, \\
&R=\left(\begin{array}{ccccccccc}
0&-d_0d_1+e_1&0&\\
0&-d_0-d_2&d_1(d_0+d_2)\\
&-1&d_1&-(d_0+d_2)(d_1+d_3)+e_3\\
&&\ddots&\ddots&\ddots
\end{array}
\right).
\end{align*} 
More explicitely, $L$ and $R$ are tridiagonal matrices with the following nonzero entries 
\begin{align*}
&L_{i,i}=d_i,\quad L_{i+1,i}=1,\quad L_{i,i+1}=-e_i,\\
&R_{2i,2i}=-\sum_{j=0}^{i}d_{2j},\quad R_{2i,2i+1}=\left(\sum_{j=0}^{i}d_{2j}\right)\left(\sum_{j=0}^{i-1}d_{2j+1}\right) ,\\
&R_{2i-1,2i-1}=\sum_{j=0}^{i-2}d_{2j+1},\quad R_{2i-1,2i}=-\left(\sum_{j=0}^{i-1}d_{2j+1}\right)\left(\sum_{j=0}^{i-1}d_{2j}\right)+e_{2i-1},\quad  R_{2i-1,2i-2}=-1,
\end{align*}
for $i=1,2,\ldots$.
\end{lemma} 

This above results can be interpreted as an isospectral flow for the generalized eigenvalue problem \cite{vinet1998integrable,watkins1988} and 
a formal compatibility check leads to  the generalized Lax equation.
\begin{corollary}
The time evolution of $L$ takes the form
\begin{align}
\dot L=RFL-LF R. \label{eq:inf_fst_lax}
\end{align}
\end{corollary} 
\begin{proof} 
Indeed, it immediately follows from \eqref{lax_inf_fst} that 
\begin{align}
zE\dot \Psi =\dot L \Psi +L \dot \Psi=\dot L \Psi +LF R\Psi.\label{lax_inf_fst_1}
\end{align}
However, by taking advantage of special properties of 
$R,E$ and $F$, we have 
$$E\dot \Psi =R\Psi,\qquad z\Psi-zD\Psi= F L \Psi,\qquad D=\text{diag}{(1,0,0\cdots)}$$
which leads to 
\begin{align}
zE\dot \Psi =z R\Psi=R F L \Psi+zRD\Psi=R F L \Psi.\label{lax_inf_fst_2}
\end{align}
Then, the generalized Lax matrix equation \eqref{eq:inf_fst_lax} follows from combining \eqref{lax_inf_fst_1} and \eqref{lax_inf_fst_2}.
\end{proof} 

Finally, for the record, we present an explicit content of  the generalized Lax matrix equation \eqref{eq:inf_fst_lax} 
\begin{subequations}\label{eq:inf_fst}
\begin{align}
&\dot d_{2k-1}=-d_{2k-1}\left(\sum_{j=0}^{k-1}d_{2j}\right)\left(2\sum_{j=0}^{k-2}d_{2j+1}+d_{2k-1}\right)-e_{2k-2}\left(\sum_{j=0}^{k-2}d_{2j+1}\right)+e_{2k-1}\left(\sum_{j=0}^{k-1}d_{2j+1}\right),\\
&\dot d_{2k}=d_{2k}\left(\sum_{j=0}^{k-1}d_{2j+1}\right)\left(2\sum_{j=0}^{k-1}d_{2j}+d_{2k}\right)+e_{2k-1}\left(\sum_{j=0}^{k-1}d_{2j}\right)-e_{2k}\left(\sum_{j=0}^{k}d_{2j}\right),
\end{align}
\end{subequations}
valid for $k=1,2\ldots.$ Again, note that $d_0=\int d\nu(z)$ and any void 
sum is taken to be zero.

\begin{remark}
Within the context of isospectral flows connected with generalized eigenvalue problems \cite{watkins1988}, to our knowledge, there are only two known examples, i.e. the relativistic Toda chain  \cite{kharchev1997faces} and the R-I chain \cite{vinet1998integrable}. We have here another concrete example of such an isospectral flow.
\end{remark}

\section{Finite FST lattice} \label{sec:finite_fst}
In this section, we consider the finite discrete measure 
\begin{align}\label{dis_measure}
d\nu(z;t)=\sum_{i=1}^Kb_i(t)\delta_{\zeta_i}dz
\end{align}
where 
$$b_i(t)=b_i(0)e^{\zeta_it},\qquad  0<\zeta_1<\zeta_2<\cdots<\zeta_K.$$
In this case, we have 
$$ \det\left(C_k^{(l,p)}(\nu,\bf{e})\right) \neq 0,\quad 0\leq l\leq K,\qquad\qquad \det\left(C_k^{(l,p)}(\nu,\bf{e})\right) =0,\quad  l>K,$$
for $0\leq p$ whenever $\, p+l-1\leq k-l$ by virtue of \eqref{eq:detCSV1}. Consequently, in this degenerate case, $T_k(z)$ for $k>2K$ are not well defined. 
So we focus on the finite number of TST polynomials $\left\{T_k(z)\right\}_{k=0}^{2K-1}$ by setting up a generalized eigenvalue problem 
whose characteristic polynomial is proportional to $T_{2K}$. Before this is 
done we prove some basic facts about TST polynomials associated 
to \eqref{dis_measure}.  

\begin{lemma} \label{lem:id_nu}
Suppose $d\nu(z)$ is a finite discrete measure \eqref{dis_measure}.  
Then 
\begin{align}
\int d\nu(z)=(-1)^{K+1}\frac{\det \left(C_{2K}^{\left(K+1,0\right)}\right)}{\det \left(C_{2K}^{\left(K,0\right)}\right)}= -\sum_{j=1}^Kd_{2j}.  
\end{align}
\end{lemma}
\begin{proof}
By Corollary \ref{coro_sum}, it suffices to prove 
\begin{align*}
0=&(-1)^K\int d\nu(z)\det \left(C_{2K}^{\left(K,0\right)}\right)+\det \left(C_{2K}^{\left(K+1,0\right)}\right)\\
=&\det\left(
\begin{array}{cccccccc}
   0&0&\cdots&\int d\nu(z)&0&0&\cdots&1\\
    V(e_1)&e_1V(e_1)&\cdots&e_1^{K}V(e_1)&1&e_1&\cdots&e_1^{{K-1}}\\
    V(e_2)&e_2V(e_2)&\cdots&e_2^{K}V(e_2)&1&e_2&\cdots&e_2^{K-1}\\
    \vdots&\vdots&\ddots&\vdots&\vdots&\vdots&\ddots&\vdots\\
    V(e_{2K})&e_{2K}V(e_{2K})&\cdots&e_{2K}^{K}V(e_{2K})&1&e_{2K}&\cdots&e_{2K}^{K-1}
  \end{array}
\right)
\end{align*}
when the measure $d\nu$  is given by the discrete measure \eqref{dis_measure}.
Let us call the above determinant  $\mathscr{R}$ and let us implement the following steps.  First, we rewrite the determinant in terms of a multiple integral 
\begin{align*}
\mathscr{R}=\int
\det \begin{pmatrix} 
0&0&\cdots&1&0&0&\cdots&1\\
\frac{1}{e_1+x_1}&\frac{e_1}{e_1+x_2}&\cdots&\frac{e_1^{K}}{e_1+x_{K+1}}&1&e_1&\dots&e_1^{K-1}\\
\frac{1}{e_2+x_1}&\frac{e_2}{e_2+x_2}&\cdots&\frac{e_2^{K}}{e_2+x_{K+1}}&1&e_2&\dots&e_2^{K-1}\\
\vdots&\vdots&\ddots&\vdots&\vdots&\vdots&\ddots&\vdots\\
\frac{1}{e_{2K}+x_1}&\frac{e_{2K}}{e_{2K}+x_2}&\cdots&\frac{e_{2K}^{K}}{e_{2K}+x_{K+1}}&1&e_{2K}&\dots&e_{2K}^{K-1}\
\end{pmatrix} d\nu(x_1)d\nu(x_2)\cdots d\nu(x_{K+1}).
\end{align*}
Then, for a fixed $j=2,3,\ldots, K+1$, we add the $(j+K-i)$-th column multiplied by $(-1)^{i+1}x_j^{i}$ for each $i=0,1,\ldots,j-2$ to the $j$-th column,  and subsequently factor$(-x_j)^{j-1}$ to obtain
\begin{align*}
\mathscr{R}&
=\int
\det \begin{pmatrix} 
0&0&\cdots&0&0&0&\cdots&1\\
\frac{1}{e_1+x_1}&\frac{1}{e_1+x_2}&\cdots&\frac{1}{e_1+x_{K+1}}&1&e_1&\dots&e_1^{K-1}\\
\frac{1}{e_2+x_1}&\frac{1}{e_2+x_2}&\cdots&\frac{1}{e_2+x_{K+1}}&1&e_2&\dots&e_2^{K-1}\\
\vdots&\vdots&\ddots&\vdots&\vdots&\vdots&\ddots&\vdots\\
\frac{1}{e_{2K}+x_1}&\frac{1}{e_{2K}+x_2}&\cdots&\frac{1}{e_{2K}+x_{K+1}}&1&e_{2K}&\dots&e_{2K}^{K-1}\
\end{pmatrix} \prod_{j=1}^{K+1}x_j^{j-1}d\nu(x_j),
\end{align*}
from which and the fact that $d\nu(z)$ has finite support with $K$ atoms, it immediately follows that $\mathscr{R}$ vanishes. Thus the proof is completed.
\end{proof}

This lemma suggests a natural truncation of the infinite FST lattice \eqref{eq:inf_fst} and allows one to formulate its Lax form. Before we proceed we note that in order to terminate the recurrence \eqref{fst_3term} 
at the $2K$th step we need $z$ to be one of the zeros of the non-trivial 
polynomial $T_{2K}(z)$.  Based on the derivation in the previous section and Lemma \ref{lem:id_nu} justifying the replacement of  $d_0$ by $-\sum_{j=1}^Kd_{2j}$ we 
obtain  an analog of Lemma \ref{lem:Laxinfty} for 
the truncated system.  
\begin{lemma} \label{lem:Laxfinite}
Let $\Psi_{[2K]}=(\psi_{0}(z;t),\psi_{1}(z;t)\cdots, \psi_{2K-1}(z,t))^\top, 
 \quad F_{[2K]}=E_{[2K]}^\top$.  
Then the truncated recurrence equations \eqref{fst_3term} and their 
time evolution can be written 
\begin{align}
L_{[2K]}\Psi_{[2K]}=zE_{[2K]}\Psi_{[2K]},\qquad \dot \Psi_{[2K]}=F_{[2K]} R_{[2K]}\Psi_{[2K]},
\end{align} 
where
\begin{align*}
&L_{[2K]}=\left(\begin{array}{ccccccc}
d_1&-e_1&\\
1&d_2&-e_2\\
&\ddots&\ddots&\ddots\\
&&1&d_{2K-1}&-e_{2K-1}\\
&&&1&d_{2K}
\end{array}
\right),
\quad 
E_{[2K]}=\left(\begin{array}{ccccccc}
0&1\\
&0&1\\
&&\ddots&\ddots\\
&&&0&1\\
&&&&0
\end{array}
\right),\\
& R_{[2K]}=\left(\begin{array}{ccccccccc}
0&d_1\sum\limits_{j=1}^{K}d_{2j}+e_1&\\
0&\sum\limits_{j=2}^{K}d_{2j}&-d_1\sum\limits_{j=2}^{K}d_{2j}\\
&\ddots&\ddots&\ddots\\
&&-1&\sum_{j=0}^{K-2}d_{2j+1}&d_{2K}\sum_{j=0}^{K-1}d_{2j+1}+e_{2K-1}\\
&&&0&0
\end{array}
\right),
\end{align*} 
with the nonzero entries being given by 
\begin{align*}
&L_{i,i}=d_i,\quad L_{i+1,i}=1,\quad L_{i,i+1}=-e_i,\quad E_{i,i+1}=1,\\
&R_{2i,2i}=\sum_{j=i+1}^{K}d_{2j},\quad R_{2i,2i+1}=-\left(\sum_{j=i+1}^{K}d_{2j}\right)\left(\sum_{j=0}^{i-1}d_{2j+1}\right) ,\\
&R_{2i-1,2i-1}=\sum_{j=0}^{i-2}d_{2j+1},\quad R_{2i-1,2i}=\left(\sum_{j=0}^{i-1}d_{2j+1}\right)\left(\sum_{j=i}^{K}d_{2j}\right)+e_{2i-1},\quad  R_{2i-1,2i-2}=-1.
\end{align*}
\end{lemma} 

The compatibility condition yields the generalized Lax form
\begin{corollary} 
The time evolution of $L_{[2K]}$ takes the form
\begin{align}\label{eq:genLax}
\dot L_{[2K]}=R_{[2K]}F_{[2K]}L_{[2K]}-L_{[2K]}F_{[2K]}R_{[2K]},
\end{align}

or, equivalently, 
\begin{subequations}\label{eq:finite_fst}
\begin{align}
&\dot d_{2k-1}=d_{2k-1}\left(\sum_{j=k}^{K}d_{2j}\right)\left(2\sum_{j=0}^{k-2}d_{2j+1}+d_{2k-1}\right)-e_{2k-2}\left(\sum_{j=0}^{k-2}d_{2j+1}\right)+e_{2k-1}\left(\sum_{j=0}^{k-1}d_{2j+1}\right),\\
&\dot d_{2k}=-d_{2k}\left(\sum_{j=0}^{k-1}d_{2j+1}\right)\left(d_{2k}+2\sum_{j=k+1}^{K}d_{2j}\right)-e_{2k-1}\left(\sum_{j=k}^{K}d_{2j}\right)+e_{2k}\left(\sum_{j=k+1}^{K}d_{2j}\right),
\end{align}
\end{subequations}
for $k=1,2,\ldots,K.$
\end{corollary} 
We will refer to equations \eqref{eq:finite_fst} as the 
\textbf {\textit {finite FST lattice}}.  
The following result holds by definition
\begin{theorem}\label{th:sol_finite_fst}
The finite FST lattice \eqref{eq:finite_fst} admits the solution
 \begin{align*}
&d_{2k+1}=(-1)^k\frac{\det \left(C_{2k+1}^{\left(k,0\right)}\right)}{\det \left(C_{2k+1}^{\left(k+1,0\right)}\right)}-(-1)^{k-1}\frac{\det \left(C_{2k-1}^{\left(k-1,0\right)}\right)}{\det \left(C_{2k-1}^{\left(k,0\right)}\right)},\\
&d_{2k+2}=(-1)^{k+1}\frac{\det \left(C_{2k+2}^{\left(k+2,0\right)}\right)}{\det \left(C_{2k+2}^{\left(k+1,0\right)}\right)}-(-1)^{k}\frac{\det \left(C_{2k}^{\left(k+1,0\right)}\right)}{\det \left(C_{2k}^{\left(k,0\right)}\right)},
 \end{align*}
 where $C_{k}^{\left(l,p\right)}$ is the CSV matrix with 
 \begin{align*}
d\nu(z;t)=\sum_{i=1}^Kb_i(t)\delta_{\zeta_i}dz,\quad b_i(t)=b_i(0)e^{\zeta_it},\qquad  0<\zeta_1<\zeta_2<\cdots<\zeta_K.
\end{align*}
\end{theorem}
To conclude this section we give the following description of the 
spectrum of $L_{[2K]}\Psi_{[2K]}=zE_{[2K]}\Psi_{[2K]}$.  Clearly, the eigenvalues are automatically zeros of $T_{2K}(z)$, but one can sharpen this statement.  
\begin{lemma} \label{lem:spectrum}
Let $L_{[2K]}$ and $E_{[2K]}$ be given as in Lemma \ref{lem:Laxfinite}. Then 
\mbox{}
\begin{enumerate} 
\item 
\begin{equation*}
\det (L_{[2K]}-zE_{[2K]})=T_{2K}(z).  
\end{equation*}
\item For any $\zeta_l$ in the support of $d\nu$  given by \eqref{dis_measure}, 
$T_{2K}(\zeta_l)=0$.  
\item The polynomial $\det (L_{[2K]}-zE_{[2K]})$ is time invariant 
for the FST finite lattice.  
\end{enumerate} 
\end{lemma} 
\begin{proof} 
By inspection we see that $\det (L_{[2K]}-zE_{[2K]})$ is a monic polynomial in $z$ of degree $K$ and so is $T_{2K}(z)$.  Hence they must be equal as they have identical roots (by definition of the truncation of the recurrence relation \eqref{fst_3term}).  
The second statement can be proven by a direct computation as follows. 
Let us fix $z=\zeta_l$ where $\zeta_l$ is in the support 
of $d\nu(z)$.  Then $T_{2K}(\zeta_l)$ is proportional to the 
multiple integral 
\begin{align*}
\int
\det \begin{pmatrix} 
1&-\zeta_l&\cdots&(-\zeta_l)^K&0&0&\cdots&0\\
\frac{1}{e_1+x_1}&\frac{e_1}{e_1+x_2}&\cdots&\frac{e_1^{K}}{e_1+x_{K+1}}&1&e_1&\dots&e_1^{K-1}\\
\frac{1}{e_2+x_1}&\frac{e_2}{e_2+x_2}&\cdots&\frac{e_2^{K}}{e_2+x_{K+1}}&1&e_2&\dots&e_2^{K-1}\\
\vdots&\vdots&\ddots&\vdots&\vdots&\vdots&\ddots&\vdots\\
\frac{1}{e_{2K}+x_1}&\frac{e_{2K}}{e_{2K}+x_2}&\cdots&\frac{e_{2K}^{K}}{e_{2K}+x_{K+1}}&1&e_{2K}&\dots&e_{2K}^{K-1}\
\end{pmatrix} \prod_{j=1}^{K+1}d\nu(x_j).
\end{align*}
Now, we follow the same procedure as in the proof of 
Lemma \ref{lem:id_nu}, namely, for a fixed $j=2,3,\ldots, K+1$, we add the $(j+K-i)$-th column multiplied by $(-1)^{i+1}x_j^{i}$ for each $i=0,1,\ldots,j-2$ to the $j$-th column, obtaining 
\begin{align*}
&
=\int
\det \begin{pmatrix} 
1&-\zeta_l&\cdots&(-\zeta_l)^K&0&0&\cdots&\\
\frac{1}{e_1+x_1}&\frac{-x_2}{e_1+x_2}&\cdots&\frac{(-x_{K+1})^K}{e_1+x_{K+1}}&1&e_1&\dots&e_1^{K-1}\\
\frac{1}{e_2+x_1}&\frac{-x_2}{e_2+x_2}&\cdots&\frac{(-x_{K+1})^K}{e_2+x_{K+1}}&1&e_2&\dots&e_2^{K-1}\\
\vdots&\vdots&\ddots&\vdots&\vdots&\vdots&\ddots&\vdots\\
\frac{1}{e_{2K}+x_1}&\frac{-x_2}{e_{2K}+x_2}&\cdots&\frac{(-x_{K+1})^K}{e_{2K}+x_{K+1}}&1&e_{2K}&\dots&e_{2K}^{K-1}\
\end{pmatrix} \prod_{j=1}^{K+1}d\nu(x_j).  
\end{align*}
Since the measure $d\nu$ has $K$ atoms there are always two columns 
among the first $K+1$ columns of the integrand which 
are proportional one to another, hence $T(\zeta_l)=0$.  
Finally, since $T_{2K}(z)$ is time invariant, so is $\det (L_{[2K]}-zE_{[2K]})$ by item ($1$) above.  
\end{proof} 
\begin{remark} Note that there are $K$ nontrivial coefficients 
in the monic polynomial $\det (L_{[2K]}-zE_{[2K]})$ which has degree $K$.  
Thus we identified $K$ nontrivial constants of motion.  Since 
the Lax equation does not have a commutator form these 
constants of motion are not traces of powers of $L_{[2K]}$.  
Moreover, even though one can map the generalized Lax 
equation \eqref{eq:genLax} into a commutator equation, for example
\begin{equation} 
\frac{d}{dt} ( L_{[2K]}F_{[2K]})=\big[R_{[2K]}F_{[2K]}, L_{[2K]}F_{[2K]}\big], 
\end{equation} 
one sees that one is not gaining new constants of motion as the 
spectrum of $L_{[2K]}F_{[2K]}$ consists of $e_1,\cdots, e_{2K}$.

\end{remark}

\section{mCH peakon lattice}
This section reviews some facts about the mCH peakons, which will be used later to develop the relation with the finite FST lattice.
The mCH equation is a system 
\begin{equation}\label{eq:m1CH}
m_t+\left((u^2-u_x^2) m\right)_x=0,  \qquad 
m=u-u_{xx},
\end{equation}
which admits a special class of non-smooth solutions called peakons defined by the ansatz
\begin{equation*} 
u(x,t)=\sum_{j=1}^n m_j (t)e^{-|x-x_j(t)|}.
\end{equation*} 
It has been shown in \cite{chang2017lax} that, assuming the peakon ansatz as well as the ordering condition $x_1< x_2<\cdots<x_n$ , (\ref{eq:m1CH})  can be viewed as a distribution equation provided the ODE system
\begin{equation}\label{mCH_ode}
\dot{m}_j=0, \qquad 
\dot{x}_j=2\sum_{\substack{1\leq k\leq n,\\k\neq j}}m_jm_ke^{-|x_j-x_k|}+4\sum_{1\leq i<j<k\leq n}m_im_ke^{-|x_i-x_k|}
\end{equation}
holds.  
The mCH peakon ODE system \eqref{mCH_ode} describes actually an isospectral deformation (spectrum preserving) of the spectral problem \begin{equation} \label{dstring}
\begin{gathered}
\begin{aligned}
     q_{k}-q_{k-1}&=h_kp_{k-1}, & 1\leq k\leq n, \\
     p_{k}-p_{k-1}&=-z g_kq_{k-1},& 1\leq k\leq n,\\
      q_0=0, \quad  p_0=1&, \quad p_{n}=0,  &  
  \end{aligned}
\end{gathered}
\end{equation} 
where $g_j=m_j e^{-x_j}, \, h_j =m_j e^{x_j} $ and $z\in \mathbf{C}$ is a spectral variable. For future use note that $g_jh_j=m_j^2$ and we focus on the case when all $m_k$ are \textbf{positive} and \textbf{distinct}.
In \cite{chang2017lax}, an inverse spectral method has been formulated to solve the mCH peakon ODEs \eqref{mCH_ode} and hence \eqref{eq:m1CH}.  For our purpose, we mainly address the case of even $n$, say $n=2K$, hereafter.

\subsection{Forward and inverse problems}
In this subsection, we review some results in \cite{chang2017lax} regarding the forward and inverse spectral problems for \eqref{dstring} at the initial time $t=0$. Given the initial positions positions ordered as $x_1(0)<x_2(0)<\cdots<x_n(0)$, and \textbf{positive} and \textbf{distinct} constants $m_k$, we consider the forward spectral problems for \eqref{dstring} using the notation introduced in Section \ref{sec:notations}.

\begin{theorem}[{\cite[Corollary 2.7]{chang2017lax}}]
Consider the initial value problem
\begin{equation} \label{dstringIVP}
\begin{gathered}
\begin{aligned}
     q_{k}-q_{k-1}&=h_kp_{k-1}, & 1\leq k\leq 2K, \\
     p_{k}-p_{k-1}&=-z g_kq_{k-1},& 1\leq k\leq 2K,\\
      q_0=0, &\quad  p_0=1.& 
  \end{aligned}
\end{gathered}
\end{equation} 
Then the polynomials $p_k(z),q_k(z)$ can be explicitly expressed as 
\begin{subequations}
\begin{align}
q_k(z)&=
\sum_{j=0}^{\lfloor\frac{k-1}{2}\rfloor}\Big(\sum_{\substack{I\in \binom{[k]}{j+1}, J\in \binom{[k]}{j}\\ I<J}}
\, h_Ig_J\Big) (-z)^j,  \\
p_{k}(z)&=1+\sum_{j=1}^{\lfloor\frac{k}{2}\rfloor}\Big(\sum_{\substack{I,J \in \binom{[k]}{j}\\ I<J}} h_I g_J
\, \Big)(-z)^j.  
\end{align}
\end{subequations} 
 
\end{theorem}

The spectrum of the boundary value problem \eqref{dstring} as well as 
additional spectral data is captured by the \textit{Weyl function}
\begin{equation}\label{eq:defWeyl}
W(z)=\frac{q_{2K}(z)}{p_{2K}(z)},  
\end{equation} 
whose main properties are 
\begin{theorem}[{\cite[Theorem 3.1]{chang2017lax}}] \label{thm:W} 
Given the boundary value problem \eqref{dstring}, $W(z)$ defined by \eqref{eq:defWeyl} is a Stieltjes transform of a positive, discrete measure $d\mu$ with support 
in $\R_+$.  More precisely: 
\begin{equation}\label{weyl_ration}
W(z)=\int \frac{d\mu(x)}{x-z}, \qquad d\mu=\sum_{i=1}^{K} 
b_i \delta_{\zeta _i}, \qquad 0<\zeta_1<\dots< \zeta_{K}, \qquad 0<b_i,   \quad 1\leq j\leq K.
\end{equation}
\end{theorem} 

Conversely, an inverse problem can be fashioned with the help of  \textit {Cauchy-Jacobi interpolation problem}, known from the general multi-point Pad\'{e} approximation theory \cite{baker1996pade}. 

\begin{theorem}[{\cite[Theorem 4.20]{chang2017lax}}] 
\label{thm:detsolISP}
Given a rational function \eqref{weyl_ration},  one can uniquely determine positive constants $g_j, h_j$, $1\leq j\leq 2K$, such 
that $g_jh_j=m_j^2$ and the initial value problem: 
\begin{equation*} 
\begin{gathered}
\begin{aligned}
     q_{k}-q_{k-1}&=h_kp_{k-1}, & 1\leq k\leq 2K, \\
     p_{k}-p_{k-1}&=-z g_kq_{k-1},& 1\leq k\leq 2K,\\
      q_0=0, &\quad  p_0=1, &  
  \end{aligned}
\end{gathered}
\end{equation*} 
satisfies 
\begin{equation*}
W(z)=\frac{q_{2K}(z)}{p_{2K}(z)}.  
\end{equation*} 
The unique solution can be explicitly expressed as
\begin{subequations}
\begin{align}
&&g_{k'}&=\frac{(-1)^{\frac{k-1}{2}}\det \left(C_k^{(\frac{k-1}{2},1)}\right)\det \left(C_{k-1}^{(\frac{k-1}{2},1)}\right)}
{\mathbf{e}_{[1,k]}\det \left(C_k^{(\frac{k+1}{2},0)}\right)\det  \left(C_{k-1}^{(\frac{k-1}{2},0)}\right)},   &\text{ if k  is odd},  \label{eq:detinversegodd}\\
&&g_{k'}&= \frac{(-1)^{\frac{k}{2}}\det \left(C_k^{(\frac{k}{2},1)}\right)\det  \left(C_{k-1}^{(\frac k2 -1,1)}\right)}
{\mathbf{e}_{[1,k]}\det \left(C_k^{(\frac k2,0)}\right)\det  \left(C_{k-1}^{(\frac{k}{2},0)}\right)}, &\text{ if k is even}. \label{eq:detinversegeven}
\end{align}
\end{subequations}
Likewise, 
\begin{subequations}
\begin{align}
h_{k'}&=\frac
{\mathbf{e}_{[1,k-1]}\det \left(C_k^{(\frac{k+1}{2},0)}\right)\det  \left(C_{k-1}^{(\frac{k-1}{2},0)}\right)}{(-1)^{\frac{k-1}{2}}\det \left(C_k^{(\frac{k-1}{2},1)}\right)\det  \left(C_{k-1}^{(\frac{k-1}{2},1)}\right)},   &\text{ if k is odd},  \label{eq:detinversehodd}\\
h_{k'}&= \frac
{\mathbf{e}_{[1,k-1]}\det \left(C_k^{(\frac k2,0)}\right)\det  \left(C_{k-1}^{(\frac{k}{2},0)}\right)}{(-1)^{\frac{k}{2}}\det \left(C_k^{(\frac{k}{2},1)}\right)\det  \left(C_{k-1}^{(\frac k2 -1,1)}\right)}, &\text{ if k is even}, \label{eq:detinverseheven}
\end{align}
\end{subequations}
where $k'=2K+1-k$ for short, $ C_k^{(l,p)}$ stands for $ C_k^{(l,p)}(\mu,\bf{e})$ defined by  \eqref{CSV}, and $\mathbf{e}_{[1,k]}=e_1e_2\cdots e_k$ with $e_j=\frac{1}{m_{j'}^2}$.
\end{theorem} 

By using the relation $h_j=m_je^{x_j}$, one can finally arrive at the inverse formulae linking the spectral data with the positions of peakons, i.e. $\{b_j,\zeta_j\}_{j=1}^K\rightarrow \{x_j\}_{j=1}^{2K}.$ 

\begin{theorem} [{\cite[Theorem 4.21]{chang2017lax}}]  \label{thm:inversex}
Given positive and distinct constants $m_j$, let $\{d\mu(x;0)\}$ be the associated spectral data of the boundary value 
problem \ref{dstring} ensured by Theorem \ref{weyl_ration}.  
Then the positions $x_j(0)$ (of peakons) can be expressed in 
terms of the spectral data as: 

\begin{subequations}
\begin{align}
&&x_{k'}&=\ln \frac
{\mathbf{e}_{[1,k-1]}\det \left(C_k^{(\frac{k+1}{2},0)}\right)\det  \left(C_{k-1}^{(\frac{k-1}{2},0)}\right)}{(-1)^{\frac{k-1}{2}}m_{k'}\det \left(C_k^{(\frac{k-1}{2},1)}\right)\det  \left(C_{k-1}^{(\frac{k-1}{2},1)}\right)},   &\text{ if k is odd}, \label{eq:detinversexodd}\\
&&x_{k'}&= \ln \frac
{\mathbf{e}_{[1,k-1]}\det \left(C_k^{(\frac k2,0)}\right)\det  \left(C_{k-1}^{(\frac{k}{2},0)}\right)}{(-1)^{\frac{k}{2}}m_{k'}\det \left(C_k^{(\frac{k}{2},1)}\right)\det  \left(C_{k-1}^{(\frac k2 -1,1)}\right)}, &\text{ if k is even}, \label{eq:detinversexeven}
\end{align}
\end{subequations}
where $ C_k^{(l,p)}(\mu(x;0),\bf{e})$ is abbreviated as $ C_k^{(l,p)}$, $k'=2K-k+1, \, 1\leq k\leq 2K$.
\end{theorem} 

\subsection{Time evolution}

As previously mentioned, the mCH peakon ODE system \eqref{mCH_ode} is an isospectral evolution system. More exactly, it is shown that the spectrum of the boundary value 
problem (\ref{dstring}) are time invariant and the Weyl function evolves according to
 \begin{equation*}
\dot W=\frac 2z W-\frac{2L}{z}, 
\end{equation*}
which implies 
$$\dot b_j=\frac{2}{\zeta_j} b_j, \qquad 1\leq j\leq K.$$

Eventually, we are led to 

\begin{theorem} [{\cite[Theorem 5.1]{chang2017lax}}]
\label{thm:peakon_even}
Assuming the notation of Theorem \ref{thm:inversex}, the mCH equation \eqref{eq:m1CH} admits the multipeakon solution
\begin{equation}\label{eq:umultipeakoneven}
u(x,t)=\sum_{k=1}^{2K}m_{k'}(t)\exp(-|x-x_{k'}(t)|),
\end{equation}
where $x_{k'}$ are given by equations \eqref{eq:detinversexodd} and \eqref{eq:detinversexeven}, 
with the peakon spectral measure 
\begin{equation}\label{eq:peakon sm}
d\mu=\sum_{j=1}^{K} b_j(t) \delta_{\zeta_j}, 
\end{equation} 
$b_j(t)=b_j(0)e^{\frac{2t}{\zeta_j}}, \, 0<b_j(0)$, ordered eigenvalues $0<\zeta_1<\cdots<\zeta_K$.
\end{theorem}

\begin{remark}
The inverse procedure can not guarantee the multipeakon solution given in Theorem \ref{thm:peakon_even} to exist globally in time because the initial order $x_1(0)<x_2(0)<\cdots<x_{2K}(0)$ might cease to hold as time varies. However,  a sufficient  condition can be constructed ensuring that the peakon flow exists globally in time (see \cite[Theorem 5.6]{chang2017lax} ).
\end{remark}

\section{mCH peakon lattice vs finite FST lattice}\label{sec:mch_fst}
In this section, we establish a connection between the boundary value problem \eqref{dstring} and a finite family of FST polynomials, hence a correspondence between the mCH peakon lattice \eqref{mCH_ode} and the finite FST lattice \eqref{eq:finite_fst}.

To line up the formulae, we need the counterpart of the boundary value problem \eqref{dstring}, given by the right boundary value problem (moving from right to left rather than from left to write), with the accompanying initial value problem
\begin{equation} \label{dstring+}
\begin{gathered}
\begin{aligned}
     \hat q_{j}-\hat q_{j-1}&=-h_{j'} \hat p_{j-1}, & 1\leq j\leq 2K, \\
     \hat p_{j}-\hat p_{j-1}&=z g_{j'}\hat q_{j-1},& 1\leq j\leq 2K,\\
       \hat p_0&=0, \quad \hat q_{0}=1,    & 
  \end{aligned}
\end{gathered}
\end{equation} 
where $\hat{}$ over $qs$ or $ps$ indicates that we are moving from right to left while the prime over $j$ reflects the counting from left to right, thus 
$j'=2K-j+1$.  
The solution of this initial value problem can be elegantly formulated in terms of the Weyl function $W(z)$ \eqref{weyl_ration}, which we recall the reader accounts for the left initial value problem \eqref{dstringIVP}. To state the result, we use the notation $V(z)=W(-z)$; hereafter we only include the formula for $\hat q_k$, which can be found in Theorem 4.12 in \cite{chang2017lax}.
\begin{theorem}\label{thm:solPapproxbis}
Given a rational function $W(z)$ as specified by \eqref{weyl_ration}, as well as positive, distinct 
constants $m_1,m_2,\dots, m_{2K}$ and setting $e_{i}=~\frac{1}{m_{i'}^2}, \, 1\leq i\leq 2K$, the solution for $\hat q_k(z)$ to the initial value problem \eqref{dstring+} reads
\begin{equation}\label{eq:hatqhatQsol}
\begin{split}
\hat q_k(z)=\frac{1}{\det \left(C_k^{(\lfloor \frac k2 \rfloor,1)}\right)}&\det
\begin{bmatrix} 1&-z&\dots&(-z)^{\lfloor \frac k2 \rfloor}&0&0&\cdots&0\\
V(e_1)&e_1V(e_1)&\dots&e_1^{\lfloor \frac k2 \rfloor}V(e_1)&1&e_1&\cdots&e_1^{\lfloor \frac{k-1}{2} \rfloor}\\
\vdots&\vdots&\ddots&\vdots&\vdots&\vdots&\ddots&\vdots\\
V(e_k)&e_k V(e_k)&\dots&e_k^{\lfloor \frac k2 \rfloor}V(e_k)&1&e_k&\cdots&e_k^{\lfloor \frac{k-1}{2} \rfloor} 
\end{bmatrix}.
\end{split}
\end{equation}  
\end{theorem} 
\begin{remark}\label{rem:p_hatq}
The right boundary value problem has the same spectrum as the left boundary value problem i.e. $\hat q_{2K}(z)= p_{2K}(z)$ (see Corollary 4.9 in \cite{chang2017lax}).
\end{remark}
This expression clearly suggests a strong relation with FST polynomials because of its similarity to the form \eqref{fst_poly}: the determinantal 
portions are identical, the normalizations, on the other hand, differ.  This prompts us to consider the following change of variables (rescaling) $\hat q_k\rightarrow \hat Q_k$.  

\begin{lemma}\label{th:hatQ}
Let a new set of variables $\{\hat Q_k\}$ be defined as 
\begin{align}\label{q_hatQ}
\hat q_{2p}=\big((-1)^p\prod_{i=1}^pg_{(2i-1)' }h_{(2i)'}\big)\hat Q_{2p}, \qquad \hat q_{2p+1}=\big((-1)^ph_{1'}\prod_{i=1}^pg_{(2i)'}h_{(2i+1)'}\big) \hat Q_{2p+1}
\end{align}

Then the polynomials $\left\{\hat Q_k(z)\right\}_{k=1}^{2K}$ satisfy the three term recurrence
\begin{equation}
\hat Q_{k+1}(z)=\hat d_{k+1} \hat Q_k(z)+(z+e_k)\hat Q_{k-1}(z), 
\end{equation}
with the initial values $\hat Q_0(z)=1,\hat Q_1(z)=\hat d_{1}=h_{1'}=\frac{1}{V(e_1)}$.
Here
\begin{subequations}\label{hatQ_hat_u}
 \begin{align}
&\hat d_{2p-1}=(-1)^{p-1}\frac{\det \left(C_{2p-1}^{\left(p-1,0\right)}\right)}{\det \left(C_{2p-1}^{\left(p,0\right)}\right)}-(-1)^{p}\frac{\det \left(C_{2p-3}^{\left(p-2,0\right)}\right)}{\det \left(C_{2p-3}^{\left(p-1,0\right)}\right)}>0,\\
&\hat d_{2p}=(-1)^{p}\frac{\det \left(C_{2p}^{\left(p+1,0\right)}\right)}{\det \left(C_{2p}^{\left(p,0\right)}\right)}-(-1)^{p-1}\frac{\det \left(C_{2p-2}^{\left(p,0\right)}\right)}{\det \left(C_{2p-2}^{\left(p-1,0\right)}\right)}<0
 \end{align}
 \end{subequations}
for $p=1,2\ldots,K.$ 
\end{lemma}
\begin{proof}
Eliminating $\hat p_k$ from \eqref{dstring+}
\[
-\frac{1}{g_{k'}h_{(k+1)'}}\hat q_{k+1} +\frac{1}{g_{k'}}\left(\frac{1}{h_{(k+1)'}}+\frac{1}{h_{k'}}\right)\hat q_{k}=(z+e_k) \hat q_{k-1}.
\]
and using the definition of $\hat Q_k$ one obtains
\[
\hat Q_{k+1}=\hat d_{k+1} \hat Q_k+(z+e_k)\hat Q_{k-1},
\]
where 
\begin{subequations}\label{d_gh}
\begin{align}
\hat d_{2p}&=-\left(\frac{1}{h_{(2p)'}}+\frac{1}{h_{(2p-1)'}}\right)\frac{h_{1'}\prod_{i=1}^{p-1}g_{(2i)'}h_{(2i+1)'}}{g_{1'}\prod_{i=1}^{p-1}h_{(2i)' }g_{(2i+1)'}},\\
\hat d_{2p+1}&=\left(\frac{1}{h_{(2p+1)'}}+\frac{1}{h_{(2p)'}}\right)\frac{\prod_{i=1}^{p}g_{(2i-1)'}h_{(2i)'}}{\prod_{i=1}^{p}h_{(2i-1)' }g_{(2i)'}}.
\end{align}
\end{subequations}
By using the formula \eqref{eq:detinversehodd}-\eqref{eq:detinverseheven}, it follows from \eqref{d_gh} that
\[\hat Q_{k+1}(z)=\hat d_{k+1} \hat Q_k(z)+(z+e_k)\hat Q_{k-1}(z), 
\]
where
\begin{align*}
&\hat d_{2p}=(-1)^{p-1}\left(\frac{\det \left(C_{2p}^{\left(p,1\right)}\right)}{\det \left(C_{2p}^{\left(p,0\right)}\right)}-\frac{e_{2p-1}\det \left(C_{2p-2}^{\left(p-1,1\right)}\right)}{\det \left(C_{2p-2}^{\left(p-1,0\right)}\right)}\right)\frac{\det \left(C_{2p-1}^{\left(p,0\right)}\right)}{\det \left(C_{2p-1}^{\left(p-1,1\right)}\right)},\\
&\hat d_{2p+1}=(-1)^{p}\left(\frac{\det \left(C_{2p+1}^{\left(p,1\right)}\right)}{\det \left(C_{2p+1}^{\left(p+1,0\right)}\right)}+\frac{e_{2p}\det \left(C_{2p-1}^{\left(p-1,1\right)}\right)}{\det \left(C_{2p-1}^{\left(p,0\right)}\right)}\right)\frac{\det \left(C_{2p}^{\left(p,0\right)}\right)}{\det \left(C_{2p}^{\left(p,1\right)}\right)}.\\
\end{align*}
The expressions for $\hat d_{k}$ are actually equivalent to \eqref{hatQ_hat_u} follows from the identity
\[
T_k(0)=d_kT_{k-1}(0)+e_{k-1}T_{k-2}(0),
\]
which in turn follows from \eqref{fst_3term} and the formulae for $d_k$ stated there.

Finally, the initial values $\hat Q_0(z),\hat Q_1(z)$ obviously follow from \eqref{eq:hatqhatQsol} and \eqref{q_hatQ}. Therefore, the proof is completed.
\end{proof}
Upon comparing the formulas for $\hat d_k$ with those for $d_k$ 
stated in Property \ref{prop:3term} we arrive at the central result of the comparison 
of the mCH peakon lattice vs finite FST lattice.  
\begin{theorem} 
$\left\{\hat Q_k(z)\right\}_{k=1}^{2K}$ form a finite family of FST polynomials associated to the measure $d\mu$ given by \eqref{eq:peakon sm}.  
\end{theorem}

We can perform an analogous analysis of the initial value problem \eqref{dstringIVP} for the (left) boundary value problem \eqref{dstring}. Namely, eliminating $ p_k$, we get from \eqref{dstring}
\[
-\frac{1}{g_{k}h_{k+1}}q_{k+1} +\frac{1}{g_{k}}\left(\frac{1}{h_{k+1}}+\frac{1}{h_{k}}\right)q_{k}=(z+e_{k'})  q_{k-1}, 
\]
which upon the rescaling of variables 
\begin{align}\label{q_Q}
 q_{2k}=\frac{(-1)^{k}\prod_{i=1}^Kh_{2i-1 }g_{2i}}{\prod_{i=k+1}^Kg_{2i-1 }h_{2i}} Q_{2k}, \qquad  q_{2k+1}=\frac{(-1)^{k}\prod_{i=1}^Kh_{2i-1 }g_{2i}}{g_{2K}\prod_{i=k+1}^{K-1}g_{2i }h_{2i+1}} Q_{2k+1},
\end{align}
results in 
\[
 Q_{k+1}=\hat d_{2K+1-k}  Q_k+(z+e_{k'}) Q_{k-1},  
\]
where $\hat d_k$ is determined by $g_k,h_k$ using \eqref{d_gh}. Then, by use of Lemma \ref{th:hatQ} and considering the initial values of $Q_k$, we immediately have 
\begin{lemma}\label{th:Q}
The polynomials $\left\{Q_k(z)\right\}_{k=1}^{2K}$ satisfy the three term recurrence
\begin{equation}
 Q_{k+1}(z)=\hat d_{k'}  Q_k(z)+(z+e_{k'}) Q_{k-1}(z), 
\end{equation}
with the initial values $ Q_0(z)=0,Q_1(z)=1$.
Here
\begin{subequations}
 \begin{align}
&\hat d_{2k-1}=(-1)^{k-1}\frac{\det \left(C_{2k-1}^{\left(k-1,0\right)}\right)}{\det \left(C_{2k-1}^{\left(k,0\right)}\right)}-(-1)^{k}\frac{\det \left(C_{2k-3}^{\left(k-2,0\right)}\right)}{\det \left(C_{2k-3}^{\left(k-1,0\right)}\right)}>0,\\
&\hat d_{2k}=(-1)^{k}\frac{\det \left(C_{2k}^{\left(k+1,0\right)}\right)}{\det \left(C_{2k}^{\left(k,0\right)}\right)}-(-1)^{k-1}\frac{\det \left(C_{2k-2}^{\left(k,0\right)}\right)}{\det \left(C_{2k-2}^{\left(k-1,0\right)}\right)}<0
 \end{align}
 \end{subequations}
for $k=1,2\ldots,K.$ 
\end{lemma}
We see that the coefficients that we originally called $\{d_k\}$ and $\{e_k\}$, 
are now reflected $k\rightarrow k'=2K+1-k$.  With that proviso we have the following statement.  
\begin{theorem} 
$\left\{Q_k(z)\right\}_{k=1}^{2K}$ form a finite family of associated FST polynomials.
\end{theorem}

Based on these results, we are ready to describe the connection between the boundary value problem \eqref{dstring} and the finite family of FST polynomials $\left\{\hat Q_k(z),Q_k(z)\right\}_{k=1}^{2K}$, consequently a correspondence between the mCH peakon lattice \eqref{mCH_ode} and the finite FST lattice \eqref{eq:finite_fst}.
\begin{theorem}\label{th:weyl_eigen}
Let $e_1,e_2,\ldots e_{2K}$ be $2K$ positive and distinct constants. Given positive constants $\{h_k\}_{k=1}^{2K}$ and $\{g_k\}_{k=1}^{2K}$ satisfying $g_kh_k=\frac{1}{e_{k'}}$, then
there exists a mapping from the boundary value problem \eqref{dstring} with $n=2K$ to the generalized eigenvalue problem
\begin{align}
\hat L_{[2K]}\hat \Psi_{[2K]}=zE_{[2K]}\hat \Psi_{[2K]},
\end{align} 
where
\begin{align*}
&\hat \Psi_{[2K]}=(\hat \psi_{0}(z;t),\hat \psi_{1}(z;t)\cdots, \hat \psi_{2K-1}(z;t))^\top, \qquad \hat\psi_k(z,t)=\frac{\hat Q_k(z;t)}{\prod_{i=1}^k(z+e_i)}\\
&\hat L_{[2K]}=\left(\begin{array}{ccccccc}
\hat d_1&-e_1&\\
1&\hat d_2&-e_2\\
&\ddots&\ddots&\ddots\\
&&1&\hat d_{2K-1}&-e_{2K-1}\\
&&&1&\hat d_{2K}
\end{array}
\right),
\qquad\qquad 
E_{[2K]}=\left(\begin{array}{ccccccc}
0&1\\
&0&1\\
&&\ddots&\ddots\\
&&&0&1\\
&&&&0
\end{array}
\right).
\end{align*} 
The mapping from $\{h_k\}_{k=1}^{2K}$ (or eqivalently $\{g_k\}_{k=1}^{2K}$) to $\{\hat d_k\}_{k=1}^{2K}$ is given by
\begin{align*}
\hat d_{2p}=-\left(\frac{1}{h_{(2p)'}}+\frac{1}{h_{(2p-1)'}}\right)\frac{h_{1'}\prod_{i=1}^{p-1}g_{(2i)'}h_{(2i+1)'}}{g_{1'}\prod_{i=1}^{p-1}h_{(2i)' }g_{(2i+1)'}},\quad \hat d_{2p+1}=\left(\frac{1}{h_{(2p+1)'}}+\frac{1}{h_{(2p)'}}\right)\frac{\prod_{i=1}^{p}g_{(2i-1)'}h_{(2i)'}}{\prod_{i=1}^{p}h_{(2i-1)' }g_{(2i)'}},
\end{align*}
under which, the Weyl function defined by \eqref{eq:defWeyl} for the boundary value problem \eqref{dstring} is equivalent to  the element
in the first row and first column of $\left(zE_{[2K]}-\hat L_{[2K]}\right)^{-1}$, i.e.
$$
W(z)=\braket{1|(zE_{[2K]}-\hat L_{[2K]})^{-1}| 1}.
$$
\end{theorem}
\begin{proof}
The spectrum of the boundary value problem \eqref{dstring} is the set of the zeros of $p_{2K}(z)$, which is the polynomial generated by the initial value problem \eqref{dstringIVP}. By Remark \ref{rem:p_hatq} and \eqref{q_hatQ}, we see that the spectrum is the set of zeros of $\hat Q_{2K}(z)$ recursively generated by 
\begin{equation*}
\hat Q_{k+1}(z)=\hat d_{k+1} \hat Q_k(z)+(z+e_k)\hat Q_{k-1}(z), 
\end{equation*}
with the initial values $\hat Q_0(z)=1,\hat Q_1(z)=\hat d_{1}=h_{1'}$, where $\hat d_k$ is given by \eqref{d_gh} . By Lemma \ref{lem:spectrum} \begin{align*}
\hat Q_{2K}(z)=\det
\left(\begin{array}{ccccccc}
\hat d_1&-e_1-z&\\
1&\hat d_2&-e_2-z\\
&\ddots&\ddots&\ddots\\
&&1&\hat d_{2K-1}&-e_{2K-1}-z\\
&&&1&\hat d_{2K}
\end{array}
\right)= \det{\left(\hat L_{[2K]}-zE_{[2K]}\right)}.  
\end{align*}
Therefore, the boundary value problem \eqref{dstring} with $n=2K$ is mapped into the generalized eigenvalue problem 
\begin{align*}
\hat L_{[2K]}\hat \Psi_{[2K]}=zE_{[2K]}\hat \Psi_{[2K]}.
\end{align*} 

Regarding the Weyl function of the boundary value problem \eqref{dstring} defined by \eqref{eq:defWeyl}, we have
\begin{align*}
W(z)=\frac{q_{2K}(z)}{p_{2K}(z)}=\frac{q_{2K}(z)}{\hat q_{2K}(z)}=\frac{Q_{2K}(z)}{\hat Q_{2K}(z)}
\end{align*}
by virtue of Remark \ref{rem:p_hatq}, the relations \eqref{q_hatQ}, \eqref{q_Q}. Similar to the determinant representation for $\hat Q$, Theorem \ref{th:Q} implies that 
\begin{align*}
Q_{2K}(z)=\det
\left(\begin{array}{ccccccc}
\hat d_{1'}&-e_{2'}-z&\\
1&\hat d_{2'}&-e_{3'}-z\\
&\ddots&\ddots&\ddots\\
&&1&\hat d_{2K-1}&-e_{(2K-1)'}-z\\
&&&1&\hat d_{(2K-1)'}
\end{array}
\right),
\end{align*}
which, after going to the ``unprimed'' indices, shows that $W(z)$ equals to the element in the first row and first column of $\left(zE_{[2K]}-\hat L_{[2K]}\right)^{-1}$.
\end{proof}

When the time evolution is considered, we eventually arrive at
\begin{theorem}\label{th:mch_fst} 
Given positive and distinct constants $e_k,1\leq k\leq 2K$, let 
$$
\beta_j(0)=\sum_{i=1}^K\zeta_i(0)^jb_i(0),\qquad V_k(0)=\sum_{i=1}^K\frac{b_i(0)}{\zeta_i(0)+e_k},
$$
with
 \[
 0<\zeta_1(0)<\zeta_2(0)<\cdots<\zeta_K(0),\qquad b_i(0)>0,
 \]
For any positive integer $k$, index $p$, and $l$ such that $0\leq l\leq k$, define $\tau_k^{(l,p)}(0) $ as
\begin{equation*}
\tau_k^{(l,p)}(0)=\det\left(\begin{array}{cccccccc}
    e_1^{p}V(e_1)&e_1^{p+1} V(e_1)&\cdots&e_1^{p+l-1}V(e_1)&1&e_1&\cdots&e_1^{k-l-1}\\
    e_2^{p}V(e_2)&e_2^{p+1}V(e_2)&\cdots&e_2^{p+l-1}V(e_2)&1&e_2&\cdots&e_2^{k-l-1}\\
    \vdots&\vdots&\ddots&\vdots&\vdots&\vdots&\ddots&\vdots\\
    e_k^{p}V(e_k)&e_k^{p+1}V(e_k)&\cdots&e_k^{p+l-1}V(e_k)&1&e_k&\cdots&e_k^{k-l-1}
  \end{array}\right),
\end{equation*}
as well as $\tau_0^{(l,p)}(0)=1$ , $\tau_k^{(l,p)}(0)=0$ for $k<0$ or $l>k$.
 \begin{enumerate}
 \item Let the variables $\{x_k(0),m_k(0)\}_{k=1}^{2K}$ be defined by
\begin{align*}
 x_{k'}(0)=\ln\frac{(-1)^{\lfloor \frac{k}{2}\rfloor}\mathbf{e}_{[1,k]}\tau_{k}^{(\lfloor \frac{k+1}{2}\rfloor,0)}(0)\tau_{k-1}^{(\lfloor \frac{k}{2}\rfloor,0)}(0)}{m_{k'}\tau_{k}^{(\lfloor \frac{k}{2}\rfloor,1)}(0)\tau_{k-1}^{(\lfloor \frac{k-1}{2}\rfloor,1)}(0)},\qquad m_{k'}(0)=\frac{1}{\sqrt{e_k}},
\end{align*}
where $k'=2K+1-k,\mathbf{e}_{[1,k]}=\prod_{i=1}^ke_i.$ If $\{\zeta_i(t),b_i(t)\}_{i=1}^K$ evolve as 
 \[
 \dot \zeta_i=0,\qquad \dot b_i=\frac{2b_i}{\zeta_i},
 \]
then $\{x_k(t),m_k(t)\}_{k=1}^{2K}$  satisfy the mCH peakon ODEs \eqref{mCH_ode} with $n=2K$.
\item 
 Let the variables $\{d_k(0)\}_{k=1}^{2K}$ be defined by
 \begin{align*}
&d_{2p+1}(0)=
(-1)^p\frac{\tau_{2p+1}^{\left(p,0\right)}(0)}{\tau_{2p+1}^{\left(p+1,0\right)}(0)}-(-1)^{p-1}\frac{\tau_{2p-1}^{\left(p-1,0\right)}(0)}{\tau_{2p-1}^{\left(p,0\right)}(0)},\\
&d_{2p+2}(0)=(-1)^{p+1}\frac{\tau_{2p+2}^{\left(p+2,0\right)}(0)}{\tau_{2p+2}^{\left(p+1,0\right)}(0)}-(-1)^{p}\frac{\tau_{2p}^{\left(p+1,0\right)}(0)}{\tau_{2p}^{\left(p,0\right)}(0)}.
 \end{align*}
 If $\{\zeta_i(t),b_i(t)\}_{i=1}^K$ evolve as 
 \[
 \dot \zeta_i=0,\qquad \dot b_i=\zeta_ib_i,
 \]
 then  $\{d_k(t)\}_{k=1}^{2K}$  satisfy the finite FST lattice \eqref{eq:finite_fst}.
 \item The initial data of the mCH peakon problem $\{x_k(0),m_k(0)\}_{k=1}^{2K}$ is mapped to the initial data of 
 the FST lattice $\{d_k(0)\}_{k=1}^{2K}$  as follows 
\begin{align*}
&d_{2p}(0)=-\left(\frac{1}{h_{(2p)'}(0)}+\frac{1}{h_{(2p-1)'}(0)}\right)\frac{h_{1'}(0)\prod_{i=1}^{p-1}g_{(2i)'}(0)h_{(2i+1)'}(0)}{g_{1'}(0)\prod_{i=1}^{p-1}h_{(2i)' }(0)g_{(2i+1)'}(0)},\\
& d_{2p+1}(0)=\left(\frac{1}{h_{(2p+1)'}(0)}+\frac{1}{h_{(2p)'}(0)}\right)\frac{\prod_{i=1}^{p}g_{(2i-1)'}(0)h_{(2i)'}(0)}{\prod_{i=1}^{p}h_{(2i-1)' }(0)g_{(2i)'}(0)},
\end{align*}
where $g_j(0)=m_j(0) e^{-x_j(0)}, \, h_j(0) =m_j(0) e^{x_j(0)} $.
 \end{enumerate}
 \end{theorem}
\begin{proof}
The statement follows easily from Theorems \ref{th:sol_finite_fst}, \ref{thm:inversex}, and \ref{th:weyl_eigen}.
\end{proof}

\section{Degenerate FST lattice vs 2-mCH interlacing peakon lattice}\label{sec:fst_2mch}
In the remainder of this paper, we investigate the degenerate case of the finite FST lattice \eqref{eq:finite_fst} by choosing all $e_k$ to be the same constant $c$.

\subsection{Extreme degenerate case of the FST lattice}
When $e_k=c$, we formally obtain from \eqref{eq:finite_fst} that 
 \begin{subequations}\label{eq:d_reduce}
\begin{align}
&\dot d_{2k-1}=d_{2k-1}\left(\sum_{j=k}^{K}d_{2j}\right)\left(2\sum_{j=0}^{k-2}d_{2j+1}+d_{2k-1}\right)+c d_{2k-1},\\
&\dot d_{2k}=-d_{2k}\left(\sum_{j=0}^{k-1}d_{2j+1}\right)\left(d_{2k}+2\sum_{j=k+1}^{K}d_{2j}\right)-c d_{2k}.
\end{align}
\end{subequations}
Under the variable transformations
$$d_{2k-1}(t)=g_{2k-1}(t)e^{ct},\qquad d_{2k}(t)=g_{2k}(t)e^{-ct},$$
we immediately have 
 \begin{subequations}
\begin{align*}
&\dot g_{2k-1}=g_{2k-1}\left(\sum_{j=k}^{K}g_{2j}\right)\left(2\sum_{j=0}^{k-2}g_{2j+1}+d_{2k-1}\right),\\
&\dot g_{2k}=-g_{2k}\left(\sum_{j=0}^{k-1}g_{2j+1}\right)\left(g_{2k}+2\sum_{j=k+1}^{K}g_{2j}\right).
\end{align*}
\end{subequations}
Unexpectedly, this ODE system is, up to a scaling transformation, equivalent to (3.5)-(3.6)  in \cite{chang2018moment}, and, as shown below, may be transformed into the 2-mCH interlacing peakon ODEs.

Indeed, if we let
  \[
  p_k=\ln \frac{g_{2k-1}}{m_{2k-1}}, \qquad q_k=\ln \frac{2n_{2k}}{g_{2k}},\qquad
  \] where $m_k, n_k$ are some constants, then
  \begin{align*}
  &\dot p_k=2\sum_{i=k}^K n_{2i}e^{p_k-q_i}\left(2\ \sum_{i=1}^{k}m_{2i-1}e^{p_{i}-p_k}-m_{2k-1}\right),\\
  &\dot q_k=-2\sum_{i=1}^{k}m_{2i-1}e^{p_{i}-q_{k}}\left(2\ \sum_{i=k}^{K}n_{2i}e^{q_{k}-q_{i}}-n_{2k}\right).
\end{align*}
Let us set now 
\begin{align}\label{2mch_uv}
u(x,t)=\sum_{k=1}^K m_{2k-1}e^{-|x-p_k(t)|},\qquad v(x,t)=\sum_{k=1}^K n_{2k}e^{-|x-q_k(t)|},
\end{align}
and assume that by a choice of constants $m_k, n_k$ we can arrange for 
the variables $p_k, q_k$ to be ordered according to 
\[p_1<q_1<p_2<\cdots <p_K<q_K,\]
then the above ODE system can be rewritten as
\begin{align*}
  &\dot p_j=\left(u(p_j)-\langle u_x\rangle(p_j)\right)\left(v(p_j)+v_x(p_j)\right),\\
  &\dot q_j=\left(u(q_j)-u_x(q_j)\right)\left(v(q_j)+\langle v_x\rangle(q_j)\right),
\end{align*}
which is nothing but the 2-mCH peakon ODE system. More precisely, this ODE system is the required system ensuring that $u,v$ defined by the ansatz \eqref{2mch_uv} satisfy the PDE
\begin{align*}
  m_t&+[(u-u_x)(v+v_x)m]_x=0,\\
  n_t&+[(u-u_x)(v+v_x)n]_x=0,\\
  &m=u-u_{xx},\ \ n=v-v_{xx},
\end{align*}
in the sense of distributions as explained  in \cite{chang2016multipeakons}.

The above degeneration seems counter intuitive since the Toda-type lattices are usually viewed as positive flows in the spectral variable while peakon flows are negative flows based on previous works \cite{beals2001peakons,chang2018degasperis,chang2018application,chang2016multipeakons}.  Thus it would seem impossible to obtain a peakon flow as a reduction of a Toda-type lattice.
In order to shed some light on the above degenerate result, we shall investigate the corresponding degenerations of the FST polynomials and the solution of the FST lattice. 

\subsection{A special case of the extreme degeneration of the FST lattice}
Let us focus on the degenerate case $e_k=0$.  

When all the $e_k$ approach zero, it follows from the Heine's formula that 
\begin{align*}
\frac{\det \left(C_k^{(l,p)}(\nu,\bf{e})\right)}{\Delta_{[1,k]}(\mathbf{e})}\longrightarrow(-1)^{lp+\frac{l(l-1)}{2}} H_l^{p-k}, \qquad\qquad  \text{as} \ e_k \rightarrow0,
\end{align*}
where $H_k^l$ denotes the Hankel determinant $H_k^l=\det(A_{i+j+l})_{i,j=0}^{k-1}$ with the moments $A_k$ given by $ A_k=\int \zeta^k d\nu(\zeta)$.
This implies the limits of the FST polynomials defined by \eqref{fst_poly} exist and the degree of every FST polynomial remains the same.

It follows then from \eqref{fst_3term} that the degenerate FST polynomials satisfy the three term recurrence
$$T_{k+1}(z)=d_{k+1}T_{k}(z)+zT_{k-1}(z).$$
If we let 
$$T_k(z)=P_k(z^{-\frac{1}{2}})z^{\frac{k}{2}},$$
then $\{P_k(z)\}_{k=0}^\infty$ satisfy
$$P_{k+1}(z)=d_{k+1}zP_k(z)+P_{k-1}(z),$$
which in turn implies that 
the monic polynomials $\{S_k(z)\}_{k=0}^\infty$ defined by 
$$S_k(z)=\frac{1}{d_1d_2\cdots d_k}P_k(z)$$
satisfy
$$S_{k+1}(z)=zS_k(z)+\frac{1}{d_kd_{k+1}}S_{k-1}(z).$$
This elementary argument shows that one can associate 
 the degenerate system \eqref{eq:d_reduce} with a family of symmetric orthogonal polynomials$\{S_k(z)\}_{k=0}^\infty$ which undergo 
an isospectral deformation in the sense that the roots of 
one of them (corresponding to $T_{2K}$ in previous sections) are invariant.   

We point out that 
the map between the Kac-van Moerbeke lattice and the 2-mCH interlacing peakon lattice was established in \cite{chang2016multipeakons}. As we show above the degenerate FST lattice system \eqref{eq:d_reduce} can also be mapped to the 2-mCH interlacing peakon lattice which clearly suggests a close connection between the degenerate FST lattice and the Kac-Moerbeke lattice which merits further studies.  

To get further insight into the 
``transmutation'' of positive flows to negative flows we 
would like to offer a comment pertaining to that issue.  
Suppose we study the degeneration of the FST polynomials with the time-dependent measure based on the scheme presented in Section \ref{sec:inf_fst} and Section \ref{sec:finite_fst}. The orthogonality, after taking the limit $e_k\rightarrow 0$, gives
$$\int T_k(z;t)z^{j-k}e^{zt}d\nu(z;0)=0,\qquad j=0,1,\ldots, \lfloor \frac k2 \rfloor-1,$$
and consequently we have
$$\int P_k(z^{-\frac{1}{2}};t)z^{j-{\frac{k}{2}}}e^{zt}d\nu(z;0)=0,\qquad j=0,1,\ldots, \lfloor \frac k2 \rfloor-1,$$
which, in turn, can be written as
$$\int P_k(z;t)z^{k-2j}e^{\frac{t}{z^2}}d\nu(\frac{1}{z^2};0)=0,\qquad j=0,1,\ldots, \lfloor \frac k2 \rfloor-1,$$
finally resulting in 
$$\int S_k(z;t)z^{k-2j}e^{\frac{t}{z^2}}d\nu(\frac{1}{z^2};0)=0,\qquad j=0,1,\ldots, \lfloor \frac k2 \rfloor-1.$$
Suppose now $d\nu$ is a discrete, finite, measure.  Then 
$$\sum_{i=1}^K T_k(\zeta_i;t)\zeta_i^{j-k}e^{\zeta_it}b_i(0)=0,\qquad j=0,1,\ldots, \lfloor \frac k2 \rfloor-1.$$
and consequently we have
$$\sum_{i=1}^K P_k(\zeta_i^{-\frac{1}{2}};t)\zeta_i^{j-{\frac{k}{2}}}e^{\zeta_it}b_i(0)=0,\qquad j=0,1,\ldots, \lfloor \frac k2 \rfloor-1,$$
which, in terms of the variables $\xi_j=\zeta_j^{-\frac{1}{2}}$, can be written as
$$\sum_{i=1}^K P_k(\xi_i;t)\xi_i^{k-2j}e^{\frac{t}{\xi_i^2}}b_i(0)=0,\qquad j=0,1,\ldots, \lfloor \frac k2 \rfloor-1,$$
leading to 
$$\sum_{i=1}^K S_k(\xi_i;t)\xi_i^{k-2j}e^{\frac{t}{\xi_i^2}}b_i(0)=0,\qquad j=0,1,\ldots, \lfloor \frac k2 \rfloor-1.$$
This supports an alternative view of the motion of $\{d_k\}$ in terms of isospectral flows of symmetric orthogonal polynomials with measure with $e^{t/z^2}$ 
time dependence. 
 Combining (6.3), (6.4) in \cite{chang2016multipeakons} and (3.5), (3.6) in \cite{chang2018moment}, we find that the degenerate system \eqref{eq:d_reduce} is closely related to the 2-mCH interlacing peakon lattice. 

\subsection{The extreme case of $e_k=0$; further details. }
When all the $e_k$ approach zero, it follows from the Heine's formula that the determinant $\det \left(C_k^{(l,p)}(\nu,{\bf{e}})(t)\right)$ in the solution of the finite FST lattice \eqref{eq:finite_fst} has the limit as follows
\begin{align*}
\frac{\det \left(C_k^{(l,p)}(\nu,{\bf{e}})(t)\right)}{\Delta_{[1,k]}(\mathbf{e})}\longrightarrow(-1)^{lp+\frac{l(l-1)}{2}} H_l^{p-k}(t), \qquad\qquad  \text{as} \ e_k \rightarrow0,
\end{align*}
where $H_k^l(t)$ denotes the Hankel determinant $H_k^l(t)=\det(A_{i+j+l}(t))_{i,j=0}^{k-1}$ with the moments $A_k$ given by $ A_k(t)=\int \zeta^k e^{\zeta t} d\nu(\zeta;0)$.  We note that we will also need 
moments for negative $k$ which are well defined 
in our case since the measure has its support away from $0$.  Taking the limit we obtain  
$$
\frac{\det \left(C_k^{(l,p)}(\nu,{\bf{e}})(t)\right)}{\det \left(C_k^{(l+1,p)}(\nu,{\bf{e}})(t)\right)}\longrightarrow(-1)^{p+l} \frac{H_l^{p-k}(t)}{H_{l+1}^{p-k}(t)}, \qquad\qquad  \text{as} \ e_k \rightarrow0
$$
From solution in Theorem \ref{th:sol_finite_fst}, we obtain  
\begin{align*}
d_{2k+1}&=(-1)^k\frac{\det \left(C_{2k+1}^{\left(k,0\right)}\right)}{\det \left(C_{2k+1}^{\left(k+1,0\right)}\right)}-(-1)^{k-1}\frac{\det \left(C_{2k-1}^{\left(k-1,0\right)}\right)}{\det \left(C_{2k-1}^{\left(k,0\right)}\right)}\longrightarrow \frac{H_k^{-2k-1}}{H_{k+1}^{-2k-1}}-\frac{H_{k-1}^{-2k+1}}{H_{k}^{-2k+1}},\\
d_{2k+2}&=(-1)^{k+1}\frac{\det \left(C_{2k+2}^{\left(k+2,0\right)}\right)}{\det \left(C_{2k+2}^{\left(k+1,0\right)}\right)}-(-1)^{k}\frac{\det \left(C_{2k}^{\left(k+1,0\right)}\right)}{\det \left(C_{2k}^{\left(k,0\right)}\right)}\longrightarrow \frac{H_{k+2}^{-2k-2}}{H_{k+1}^{-2k-2}}-\frac{H_{k+1}^{-2k}}{H_{k}^{-2k}},
 \end{align*}
which suggests the following theorem. 
\begin{theorem} \label{th:dege_sol}
The degenerate system \eqref{eq:d_reduce} with $e_k=0$ admits the solution 
\begin{subequations}\label{sol:deg}
\begin{align}
&d_{2k+1}=\frac{H_k^{-2k-1}}{H_{k+1}^{-2k-1}}-\frac{H_{k-1}^{-2k+1}}{H_{k}^{-2k+1}}=\frac{(H_{k}^{-2k})^2}{H_{k+1}^{-2k-1}H_{k}^{-2k+1}},\label{sol:deg_odd}\\
&d_{2k+2}=\frac{H_{k+2}^{-2k-2}}{H_{k+1}^{-2k-2}}-\frac{H_{k+1}^{-2k}}{H_{k}^{-2k}}=-\frac{(H_{k+1}^{-2k-1})^2}{H_{k+1}^{-2k-2}H_{k}^{-2k}},\label{sol:deg_even}
 \end{align}
 \end{subequations}
 where $H_k^l(t)$ denotes the Hankel determinant $H_k^l(t)=\det(A_{i+j+l}(t))_{i,j=0}^{k-1}$ with the moments $A_k$ given by $ A_k(t)=\int \zeta^k e^{\zeta t} d\nu(\zeta;0)$.
\end{theorem}
In order to present the proof of this theorem, we need the following lemmas.
\begin{lemma} \label{lem:der_H}
The derivative of $H_k^j(t)$ with respect to $t$ satisfies
$$\dot H_k^j=G_k^j,$$
where $G_k^j$ is the determinant of the matrix obtained from $H_{k+1}^j$ by deleting the $k+1$-th row and the $k$-th column. Here we use the convention $G_k^j=0$ for $k\leq 0$.
\end{lemma}
\begin{proof}
Notice that $\dot A_k=A_{k+1}$ and the lemma follows from basic differential rules for derivatives of determinants.  
\end{proof}
\begin{lemma} \label{lem:bi_id}
The following bilinear identities hold.
\begin{align}
&H_{k+1}^lH_{k-1}^{l+2}=H_{k}^lH_{k}^{l+2}-(H_k^{l+1})^2,\label{id1}\\
&H_{k+1}^lH_{k-1}^{l+1}=G_k^{l+1}H_k^l-H_k^{l+1}G_k^l,\label{id2}\\
&G_k^lH_{k-1}^{l+1}=G_{k-1}^{l+1}H_k^l+H_k^{l+1}H_{k-1}^l.\label{id3}
\end{align}
\end{lemma}
\begin{proof}
First, we recall the well known Jacobi determinant identity \cite{aitken1959determinants}, that is, 
for any determinant $D$, 
\begin{eqnarray*}
\Big[D\left(\begin{array}{cc}
i_1 & i_2 \\
j_1 & j_2 \end{array}\right)\Big]^2&=&D\left(\begin{array}{c}
i_1  \\
j_1 \end{array}\right)\cdot D\left(\begin{array}{c}
i_2  \\
j_2 \end{array}\right)-D\left(\begin{array}{c}
i_1  \\
j_2 \end{array}\right)\cdot D\left(\begin{array}{c}
i_2  \\
j_1 \end{array}\right),
\end{eqnarray*}
where $D\left(\begin{array}{cccc}
i_1&i_2 &\cdots& i_k\\
j_1&j_2 &\cdots& j_k
\end{array}\right),\ i_1<i_2<\cdots<i_k,\ j_1<j_2<\cdots<j_k$ denotes the
determinant of the matrix obtained from $D$ by removing the rows with indices
$i_1,i_2 ,\cdots, i_k$ and the columns with indices $j_1,j_2,\cdots j_k$. 
Then the last two relations can be obtained by employing the Jacobi identity to the determinants
 \[
    D_1=\left|\begin{array}{cccc}
    A_{l}&A_{l+1}&\cdots &A_{l+k}\\
     A_{l+1}&A_{l+2}&\cdots &A_{l+k+1}\\
    \vdots&\vdots&\ddots&\vdots\\
    A_{l+k}&A_{l+k+1}&\cdots&A_{l+2k}
    \end{array}\right|,
    \]
    and 
     \[
    D_2=\left|\begin{array}{ccccc}
    A_{l}&A_{l+1}&\cdots &A_{l+k-1}&0\\
     A_{l+1}&A_{l+2}&\cdots &A_{l+k}&0\\
    \vdots&\vdots&\ddots&\vdots&\vdots\\
        A_{l+k-1}&A_{l+k}&\cdots&A_{l+2k-2}&1\\
    A_{l+k}&A_{l+k+1}&\cdots&A_{l+2k-1}&0
    \end{array}\right|,
    \]
    with
    \[i_1=1, j_1=k, i_2=j_2=k+1,\] 
    respectively.
    The first relation is a consequence of applying the Jacobi identity to the determinant $D_1$ with 
    \[i_1=j_1=1, i_2=j_2=k+1,\] 
\end{proof}
Now we are ready to present the proof of Theorem \ref{th:dege_sol}.
\begin{proof}[Proof to Theorem \ref{th:dege_sol}]
First we claim that both expressions for $d_k$ are equivalent, which follows from \eqref{id1}. The next step is to  prove the claimed form of the solution.

Indeed, on substituting the expressions \eqref{sol:deg} into the degenerate system \eqref{eq:d_reduce} with $e_k=0$ and employing the time evolution in Lemma \ref{lem:der_H}, we will see that it suffices to prove 
\begin{align*}
&2H_k^{-2k+1}G_{k-1}^{-2k+2}H_{k-1}^{-2k+3}-H_{k-1}^{-2k+2}\left( H_{k}^{-2k+1}G_{k-1}^{-2k+3}+G_{k}^{-2k+1}H_{k-1}^{-2k+3}\right)\\
=&-H_k^{-2k+2}\left(2H_{k-2}^{-2k+3}H_k^{-2k+1}+(H_{k-1}^{-2k+2})^2\right),\\
&2H_{k}^{-2k}H_{k-1}^{-2k+2}G_{k}^{-2k+1}-H_{k}^{-2k+1}\left(G_{k}^{-2k}H_{k-1}^{-2k+2}+H_{k}^{-2k}G_{k-1}^{-2k+2}\right)\\
=&H_{k-1}^{-2k+1}\left(2H_{k-1}^{-2k+2}H_{k+1}^{-2k}+(H_{k}^{-2k+1})^2\right).
\end{align*}
By rewriting the above equations, we have 
\begin{align*}
&H_k^{-2k+1}\left( G_{k-1}^{-2k+2}H_{k-1}^{-2k+3}-H_{k-1}^{-2k+2}G_{k-1}^{-2k+3}\right)+H_k^{-2k+2}\left( H_{k}^{-2k+1}H_{k-2}^{-2k+3}+(H_{k-1}^{-2k+2})^2\right)\\
&+H_{k-1}^{-2k+3}\left(  H_{k}^{-2k+1} G_{k-1}^{-2k+2}-H_{k-1}^{-2k+2}G_{k}^{-2k+1} \right)+H_k^{-2k+2} H_{k}^{-2k+1}H_{k-2}^{-2k+3}=0,\\
&H_{k-1}^{-2k+2}\left(G_{k}^{-2k}H_{k}^{-2k+1}-H_{k}^{-2k}G_{k}^{-2k+1}\right)+H_{k}^{-2k}\left(H_{k}^{-2k+1}G_{k-1}^{-2k+2}-H_{k-1}^{-2k+2}G_{k}^{-2k+1}\right)\\
&+H_{k-1}^{-2k+1}\left(H_{k-1}^{-2k+2}H_{k+1}^{-2k}+(H_{k}^{-2k+1})^2\right)+H_{k-1}^{-2k+1}H_{k-1}^{-2k+2}H_{k+1}^{-2k}=0.
\end{align*}
It is now not hard to show the validity of these relations by use of the bilinear identities in Lemma \ref{lem:bi_id}. With all these ingredients in place the proof follows.
\end{proof}
\begin{remark}
Note that Theorem \ref{th:dege_sol} gives a solution to the 2-mCH interlacing peakon ODE system, but its form is different from the one obtained from inverse spectral method in \cite{chang2016multipeakons}. 
\end{remark}

\section{Acknowledgement}
X.C. was supported in part by the National Natural Science Foundation of China (Grant Nos. 11688101, 11731014, 11701550) and the Youth Innovation Promotion Association CAS. X.H. was supported in part by the National Natural Science Foundation of China (Grant Nos. 11931017 and 11871336).
J.S. was supported in part by the Natural Sciences and Engineering Research Council
of Canada. A.Z. was supported in part  by the National Natural Science Foundation of China (Grant No.11771015). 

%
%

\begin{appendix}

\end{appendix}

\small
\bibliographystyle{abbrv}
\bibliographystyle{plain}

\def\cydot{\leavevmode\raise.4ex\hbox{.}}
  \def\cydot{\leavevmode\raise.4ex\hbox{.}} \def\cprime{$'$}

\end{document}